\documentclass{article}
\usepackage[english]{babel}
\usepackage{amsfonts,amsmath,amssymb}
\usepackage{graphicx}

\usepackage[letterpaper,top=2cm,bottom=2cm,left=3cm,right=3cm,marginparwidth=1.75cm]{geometry}
\usepackage{enumerate}  
\usepackage{tikz}
\usepackage{bbold}
\usepackage{amsmath}
\usepackage{graphicx}
\usepackage{accents}
\usepackage[colorlinks=true, allcolors=blue]{hyperref}
\usepackage[linesnumbered,ruled,vlined]{algorithm2e}
\usepackage[skins]{tcolorbox}
\SetKwInput{KwInput}{Input}               
\SetKwInput{KwOutput}{Output} 
\usepackage{bm}
\usepackage{lipsum}
\usepackage{enumitem}
\setlist[enumerate]{parsep=4pt}
\usepackage[font=small]{caption}

\usepackage{amsthm}
\usepackage{natbib}
\setcitestyle{authoryear}

\usepackage[normalem]{ulem}

\newtheorem{theorem}{Theorem}

\newtheorem{lemma}{Lemma}
\newtheorem{assumption}{Assumption}
\newtheorem{proposition}{Proposition}
\newtheorem{remark}{Remark}
\usepackage{mathtools}  
\mathtoolsset{showonlyrefs}  
\usepackage{xcolor}
\DeclareMathOperator*{\argmin}{argmin}

\usepackage{scalerel,stackengine}
\usepackage{tikz}
\usepackage{pgfplots}
\stackMath
\newcommand\reallywidehat[1]{%
\savestack{\tmpbox}{\stretchto{%
  \scaleto{%
    \scalerel*[\widthof{\ensuremath{#1}}]{\kern-.6pt\bigwedge\kern-.6pt}%
    {\rule[-\textheight/2]{1ex}{\textheight}}%
  }{\textheight}%
}{0.5ex}}%
\stackon[1pt]{#1}{\tmpbox}%
}
\title{Optimal Empirical Risk Minimization under Temporal Distribution Shifts}
\author{Yujin Jeong, Ramesh Johari, Dominik Rothenh\"ausler, Emily Fox}

\begin{document}

\maketitle

\begin{abstract}
Temporal distribution shifts pose a key challenge for machine learning models trained and deployed in dynamically evolving environments. This paper introduces \textbf{RIDER} (RIsk minimization under Dynamically Evolving Regimes) which derives optimally-weighted empirical risk minimization procedures under temporal distribution shifts. 
Our approach is theoretically grounded in the random distribution shift model, where random shifts arise as a superposition of numerous unpredictable changes in the data-generating process. We show that common weighting schemes, such as pooling all data, exponentially weighting data, and using only the most recent data, emerge naturally as special cases in our framework. We demonstrate that RIDER consistently improves out-of-sample predictive performance when applied as a fine-tuning step on the Yearbook dataset, across a range of benchmark methods in Wild-Time. Moreover, we show that RIDER outperforms standard weighting strategies in two other real-world tasks: predicting stock market volatility and forecasting ride durations in NYC taxi data.
\end{abstract}

\section{Introduction}

Temporal distribution shifts pose a fundamental challenge for machine learning models trained and deployed in dynamic environments. These shifts occur when the data-generating distribution evolves over time due to a wide range of factors, including shifting societal trends, evolving language, news events, emerging ideas and techniques, advancements in medical practice, climate change, and fluctuations in infectious disease prevalence. Recent studies have demonstrated the impact of temporal distribution shifts across diverse domains, such as large language models, drug discovery, and visual recognition 
\citep{yao2023wildtime, lazaridou2021mind, huang2021therapeutics, guo2022evaluation}.

In this paper, we introduce \textbf{RIDER} (RIsk minimization under Dynamically Evolving Regimes), a principled approach for optimally weighted empirical risk minimization under temporal distribution shifts.  
Our method is theoretically grounded in the framework of random distribution shifts, which models the cumulative effect of numerous small, stochastic changes in the data-generating process \citep{jeong2024oodgeneralization,bansak2024learning}. The idea is that, even after accounting for predictable sources of temporal variation---such as trends, seasonality or identifiable structural changes---residual shifts often persist. For example, in financial markets, while some sources of abrupt distribution shifts such as market announcements or even seasonal trends are predictable, the residual market behavior is also influenced by countless smaller idiosyncratic perturbations. Once predictable non-stationarities are accounted for, these residual shifts can often be effectively modeled as random. 

Our work is most closely related to \cite{jeong2024oodgeneralization}, but differs in a key respect: we derive an out-of-distribution (OOD) generalization bound that explicitly accounts for the trade-off between sampling variability and distributional variability (analogous to trade-off between forgetting and retention). This trade-off is essential in our setting, where the most recent time period—although most relevant for prediction—often contains relatively few observations. Without incorporating this trade-off, the method in \cite{jeong2024oodgeneralization} would tend to heavily weight recent data, increasing the risk of overfitting or unstable estimates. In addition, \cite{jeong2024oodgeneralization} addresses the broader problem of weighted empirical risk minimization, whereas our approach incorporates time-series-specific considerations. Weight estimation is improved by leveraging stationarity. We model the weights using time series techniques to derive explicit parametric forms for downweighting strategies, which can be subsequently optimized via cross-validation.

Figure~\ref{fig:diagram} illustrates the random temporal distribution shift model. 
At each time point, we collect a training dataset consisting of features and corresponding outcomes, aiming to build a predictive model that generalizes effectively to future data. The underlying data distribution evolves over time. To address the challenge of temporal shifts, we employ a weighted empirical risk minimization framework, where we assign different weights to each past training dataset. These weights should reflect how informative each past dataset is likely to be for future predictions.  

In this paper, we develop a novel approach to predictive modeling that addresses residual, random temporal distribution shifts. 
Our key contributions are as follows:
\begin{enumerate}
\item \textbf{Weighted ERM: Formulation and analysis.}  We formulate the prediction problem of interest as a problem of choosing {\em weights} for ERM over the distributions in past time steps to predict outcomes at the next time step.  Informally, distributions from more (resp., less) similar time periods should be weighted more (resp., less).  
We develop a {\em principled method} -- RIDER -- for estimation of these weights from data, and establish its theoretical out-of-distribution performance under random temporal shifts.

\item \textbf{Parametric form of optimal weights.} We show that, in special cases, the optimal weighting admits a parametric form that naturally connects to commonly used practices, including: pooling data equally; exponentially down-weighting older data; or using only the most recent data.

\item \textbf{Empirical performance.}  The method we develop can serve as a complementary tool to existing approaches to account for residual temporal shifts.
We demonstrate that our method can be applied as a fine-tuning step on top of a wide range of approaches in the Wild-Time benchmark \citep{yao2023wildtime}, consistently improving performance across methods.

We also show that RIDER can be successfully used on its own.  In particular, we show our method achieves strong performance in two real-world applications: (1) predicting volatility in financial time series data; and (2) forecasting ride duration in New York City taxi data.
\end{enumerate}

\section{Model and Problem Definition}\label{sec:problem-definition}

\subsection{Dynamically Evolving Regimes as Distribution Shifts} 

Consider a sequential supervised learning problem where, at each time point $t$, we have $n_t$ i.i.d.~samples from the distribution, $\mathbb{P}^t$: $\mathcal{D}_t = \{(X_{ti}, Y_{ti})_{i=1}^{n_t}\}$. 
Our goal is to learn a mapping $f(\cdot; \theta_t): X_t \xrightarrow[]{}Y_t$, where $\theta_t$ is an unknown parameter of the model that must be estimated from data. Having collected data sets over $t = 1, \dots, T$,  we aim to minimize the one-step-ahead out-of-distribution risk,
\begin{equation*}
    R_{T+1}(\theta) = \mathbb{E}^{T+1}[\mathcal{L}(Y, f(X;\theta))],
\end{equation*}
where $\mathbb{E}^{T+1}$ denotes the expectation under distribution $\mathbb{P}^{T+1}$, and $\mathcal{L}$ is a loss function. 

The key challenge in a time series setting arises from temporal distribution shifts; in particular, 
the distributions $\mathbb{P}^1, \dots, \mathbb{P}^T, \mathbb{P}^{T+1}$ will typically differ. Often, distributions closer in time may be more similar than those further apart.

We adopt an Empirical Risk Minimization (ERM) approach, where the parameters $\theta_{T+1}$ are estimated by minimizing an empirical risk estimator $\hat{R}_{T+1}(\theta)$. To account for potential distribution shifts, we propose using a weighted ERM approach:
\begin{equation}\label{eq:werm}
     \hat{R}_{T+1}(\theta) = \sum_{k=1}^K \beta_k\hat{\mathbb{E}}^{T+1-k}[\mathcal{L}(Y, f(X;\theta))],
\end{equation}
where $K$ is a fixed window size and $\hat{\mathbb{E}}^t$ represents the empirical mean on the dataset $\mathcal{D}_t$. The weights $\beta \in\mathbb{R}^K$ should capture similarity between the previous $K$ data distributions and the unobserved target distribution. Naive approaches include pooling all data equally ($\beta_k = 1/K$ for all $k$); down-weighting older data (e.g., $\beta_k = \theta^{k-1}/\sum_{k=1}^K \theta^{k-1}$), or using only the most recent data ({$\beta_1 = 1$; $\beta_2,\ldots, \beta_{K} = 0$). Throughout the paper, we discuss how to find optimal weights $\beta$. 

\subsection{Random Distribution Shifts} \label{sec:model}

In this section, we discuss the random distribution shift model developed in \cite{jeong2024oodgeneralization} that forms the theoretical foundation for our approach. In contrast to \cite{jeong2024oodgeneralization}, we consider the setting where both sampling variability and distributional variability are of the same order. Furthermore, we develop our methodology for the case of temporal distribution shifts.

Distribution shift models often impose some structure on the likelihood ratio between the new and original distribution. For example, under {\em covariate shift} \citep{shimodaira2000improving}, the likelihood ratio is a function of observed covariates: 
\begin{equation*}
     \frac{d \mathbb{P}^{t+1}}{d \mathbb{P}^t}(x,y) = w(x),
\end{equation*}
for some weight function $w(x) \ge 0$.  
At a high level, this assumes that covariates \emph{explain away} the distribution shift. 
In {\em domain invariant representation learning} \citep{ganin2015unsupervised}, one often aims to find a representation $\phi(X)$ that is predictive of the target, such that $\mathbb{P}^{t}(\phi(X) = \cdot) = \mathbb{P}^{s}(\phi(X) = \cdot)$, which corresponds to a constant likelihood ratio. 

The random distribution shift model also models distribution shifts via an assumption on the likelihood ratio.  However, instead of modeling the likelihood ratio as a function of observed variables, it assumes that the distribution shift arises through the combination of unpredictable, random changes that affect both the covariates and the outcomes.  In particular, for each $x,y$, we assume that 
\begin{equation*}
    \frac{d \mathbb{P}^{t}}{d \mathbb{P}^0}(x,y) \sim W_j^t,
\end{equation*}
for some random weight $W_j^t \ge 0$ with $E[W_j^t] = 1$. We make this more rigorous below.

Without loss of generality, we construct distributional perturbations for uniform distributions on $[0,1]$. This is justified by a result from probability theory \citep{dudley2018real} that any random variable or random vector $D$ on a finite or countably infinite dimensional probability space can be written as a measurable function $D \stackrel{d}{=} h(U)$, where $U$ is a uniform random variable on $[0,1]$. Through this transformation $h$, we can generalize the construction via: for any measurable set $\bullet$, 
\begin{equation*}
        \mathbb{P}^{t}(D \in \bullet) = \mathbb{P}^{t}(h(U) \in \bullet).
\end{equation*}

Let $\mathbb{P}^0$ denote the unknown common parent distribution. We now construct the perturbed distribution $\mathbb{P}^{t}$ for a random variable that follows a uniform distribution under $\mathbb{P}^0$. Denote $m$ bins $I_j = [(j-1)/m, j/m]$ for $j = 1, \dots, m$. Let $W^t_1, \dots, W^t_m$ be i.i.d.\ positive random variables with finite variance; without loss of generality, we assume $E[W_j^t] = 1$.
We assume the randomly perturbed distribution $\mathbb{P}^{t}$ is obtained from $\{ W_j^t \}_{j=1}^m$ as follows: for $u \in I_j$,
\begin{equation*}
    \frac{d \mathbb{P}^{t}}{d \mathbb{P}^0}(u) \sim W^t_j.
\end{equation*}
Assuming that the weights $W^t_1, \dots, W^t_m$ are i.i.d.\ is a strong assumption; empirical evidence supporting a similar distribution shift model in the context of replication studies is provided in \cite{jin2024beyond}.
After generating a series of perturbed distributions $\mathbb{P}^1 ,\dots, \mathbb{P}^{T+1}$, we draw $n_t$ i.i.d. samples $\mathcal{D}_t = \{D_{ti}\}_{i=1}^{n_t}$ from $\mathbb{P}^{t}$ for each $t= 1, \dots, T$, conditionally on $W = \{\{W_j^t\}_{j=1}^m\}_{t=1}^{T}$.

\paragraph{Summary of Notation.} Let $\mathbb{P}^0$ be the parent distribution on $\mathcal{D}$, and $\mathbb{P}^t$ be the perturbed distribution on $\mathcal{D}$ conditioned on $(W_j^t)_{j = 1}^m$. Denote $P$ as the marginal distribution of $\{\{D_{ti}\}_{i=1}^{n_t}, \{W_j^t\}_{j=1}^m\}\}_{t=1}^{T+1}$. We draw an i.i.d. sample $\mathcal{D}_{t} = \{D_{ti}\}_{i=1}^{n_t}$ from $\mathbb{P}^t$, conditioned on $(W_j^t)_{j = 1}^m$. 
We denote $\mathbb{E}^0$ as the expectation under $\mathbb{P}^0$, $\mathbb{E}^{t}$ as the expectation under $\mathbb{P}^{t}$, and $\hat{\mathbb{E}}^t$ as the sample average of $\mathcal{D}_t$: $\hat{\mathbb{E}}^t[\phi(D)] = \sum_{i=1}^{n_t} \phi(D_{ti})/n_t$. 
We write $\text{Var}_{\mathbb{P}^0}$ for the variance under $\mathbb{P}^0$ and $\text{Var}_{\mathbb{P}^t}$ for the variance under $\mathbb{P}^t$ conditioned on $(W_j^t)_{j =1}^m$.

\begin{figure*}[t]
    \centering
    \includegraphics[width=0.7\linewidth]{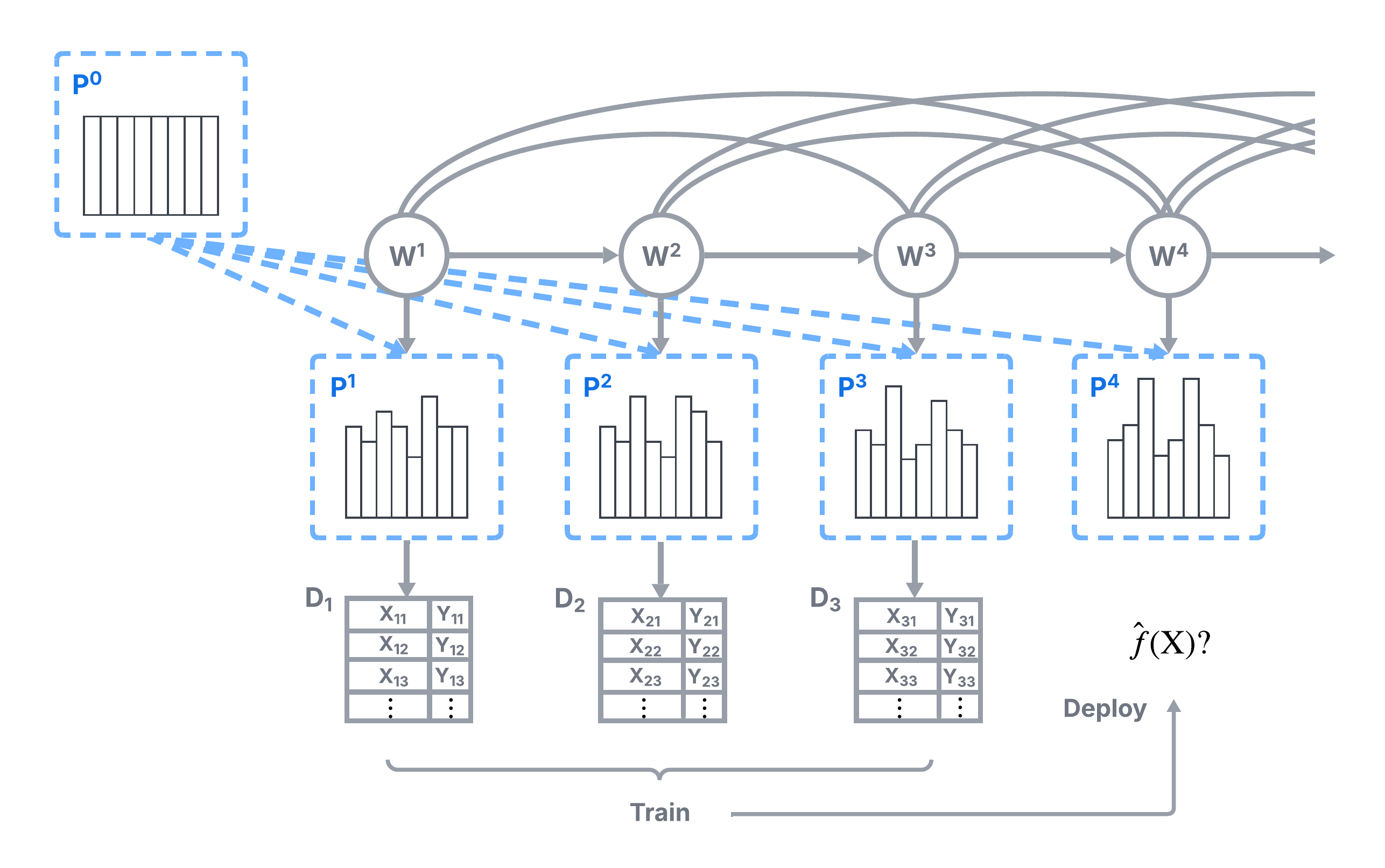}
    \caption{Diagram illustrating the data-generating process under the random temporal distribution shift model where random weights follow an AR(3) process.  }
    \label{fig:diagram}
\end{figure*}

\subsection{Optimal Weights for Weighted ERM}

We now analyze the asymptotic behavior of weighted ERM under the random distribution shift model. 
For any fixed non-negative weights $\beta$, let
\begin{equation*}
    \hat{\theta}_{T+1}^{\beta} = \argmin_{\theta \in \Theta} \sum_{k=1}^K \beta_k \hat{\mathbb{E}}^{T+1-k}[\mathcal{L}(Y, f(X; \theta))].
\end{equation*}
With $D = (X,Y)$, we define $L(\theta, D) = \mathcal{L}(Y, f(X; \theta))$.
We focus on the excess risk:
\begin{equation*}
    \mathbb{E}^{T+1}[L(\hat \theta^{\beta}_{T+1},D)]  -\mathbb{E}^{T+1}[L( \theta_{T+1},D)].
\end{equation*}
We first state the necessary regularity conditions.

\begin{assumption}[Regularity conditions]\label{ass:reg}
     For each $\theta$ in an open subset of $\Theta$, let $\theta \mapsto \partial_\theta L(\theta, d)$ be twice continuously differentiable in $\theta$ for every $d$. Assume that the matrix $\mathbb{E}^0[\partial_\theta^2 L(\theta, D)]$ exists and is nonsingular. Assume that the third partial derivatives of $\theta \mapsto L(\theta,D)$ are dominated by a fixed function $h(D)$ for every $\theta$ in a neighborhood of $\theta^0:= \argmin_{\theta \in \Theta}\mathbb{E}^0[L(\theta, D)]$. We assume that $\partial_\theta L(\theta^0,D)$, $\partial_\theta^2 L(\theta^0,D)$ and $h(D)$ are square-integrable under $\mathbb{P}^0$. 
\end{assumption}

\begin{assumption}[Sampling uncertainty and distributional uncertainty are of the same order]\label{ass:same-order} For a considered time point $t$, let $N = \sum_{t'=t-K}^{t-1}n_{t'}$ be the total sample size over the $K$-window of datasets. Assume that there exists a sequence $m = m(N)$ such that
    \begin{equation*}
        (r_{t-1}, \dots, r_{t-K})  = \lim_{N\xrightarrow[]{}\infty} \left(\frac{m}{n_{t-1}}, \dots, \frac{m}{n_{t-K}}\right) \in (0, \infty).
    \end{equation*}   
\end{assumption}

\begin{theorem}[Out-of-distribution error of weighted ERM]\label{thm:ood-error}
    Assume the distributional perturbation model mentioned in Section~\ref{sec:model}.  
    Suppose Assumption~\ref{ass:reg}, Assumption~\ref{ass:same-order} with $t = T+1$, and that $\hat{\theta}_{T+1}^{\beta} - \theta^0 = o_P(1)$, $ \theta_{T+1} - \theta^0 = o_P(1)$. Then,
\begin{equation*}
 \sqrt{m(N)} \left( \mathbb{E}^{T+1}[L( \hat{\theta}_{T+1}^{\beta}, D)]  -\mathbb{E}^{T+1}[L( \theta_{T+1}, D)] \right) \stackrel{d}{\rightarrow} \mathcal{E}
\end{equation*}
for some random variable $\mathcal{E}$, whose mean $\mathbb{E}[\mathcal{E}]$ can be decomposed into
\begin{align*}
    & \Big(\underbrace{E[(W^{T+1} - \sum_{k=1}^K \beta_k W^{T+1-k})^2]}_{\substack{\text{residual distributional} \\ \text{variation}}} + \underbrace{\sum_{k=1}^K \beta_k^2 r_{T+1-k}}_{\substack{\text{sampling} \\ \text{variation}}} \Big)\cdot \underbrace{\mu}_{\substack{\text{problem} \\\text{difficulty}}}, 
\end{align*}
where $\mu = \mathrm{Tr}( \mathbb{E}^0[\partial_\theta^2 L(\theta^0, D)]^{-1} \mathrm{Var}_{\mathbb{P}^0}(  \partial_\theta L(\theta^0, D)) ).$
\end{theorem}

The proof can be found in the Appendix, Section~\ref{sec:proof-of-ood-error}. The consistency assumptions in Theorem~\ref{thm:ood-error}; $\hat{\theta}_{T+1}^{\beta} - \theta^0 = o_P(1)$, $ \theta_{T+1} - \theta^0 = o_P(1)$ are proved in the Appendix, Section~\ref{sec:consistency}.

In contrast to \citep{jeong2024oodgeneralization}, we examine regimes where sampling uncertainty and distributional uncertainty are of the same order. This allows us to trade off distributional variation and sampling variation. The balance between these variations is determined by $r$. 
As $r \rightarrow \infty$, sampling uncertainty dominates and we assign weights that minimize the sampling variation, while as $r \rightarrow 0$, distributional uncertainty becomes more significant and we assign weights that minimize the residual distribution variation.

Let us define the distributional variation as 
\begin{equation*}
    \delta^2_t(\beta) = E[(W^{t} - \sum_{k=1}^K \beta_k W^{t-k})^2].
\end{equation*}
From Theorem~\ref{thm:ood-error}, the optimal weight of weighted ERM for the target time point $T+1$ is defined as 
\begin{align}\label{eq:opt-weights}
    \beta^{*}_{T+1} = \argmin_{\substack{\beta:\beta^{\intercal} 1 = 1 \\\beta \geq 0}} \left(\delta^2_{T+1}(\beta) + \sum_{k=1}^K \beta_k^2 r_{T+1-k}\right).
\end{align}

\subsection{Time Series Modeling of Distribution Shifts} \label{sec:time-series-modeling}

To model how data distributions evolve over time, we model the temporal structure of the random weights $(W^1$, $W^2$, $\dots$, $W^t$, $\dots)$ rather than directly modeling the data. This approach provides a unified framework for handling distribution shifts across different data modalities, including both discrete and continuous sample spaces.

There may be various ways to model the temporal dependencies in random weights. Throughout the paper, we focus on weakly stationary random weights. Specifically, we assume the covariance between any two random weights depends only on the time difference between them: $\text{Cov}(W^t, W^{t+h}) = \rho(h)$ where $\rho$ is their autocovariance function. Note that with weakly stationary random weights,
\begin{equation*}
    \delta^2(\beta) := \delta^2_t(\beta) \quad \forall t. 
\end{equation*}

For example, the random weights might obey the Autoregressive Moving Average (ARMA($p, q$)) model:
\begin{equation*}
    W_j^t =  \sum_{i=1}^p \varphi_i W_j^{t-i} + \sum_{i=1}^q \alpha_i \epsilon_j^{t-i} + \epsilon_j^t.
\end{equation*}
Note that $\alpha_1, \dots, \alpha_q, \varphi_1, \dots, \varphi_p \geq 0$ and $\epsilon_j^t$ are independent positive error terms. The model reduces to a moving average (MA($q$)) process when $p = 0$, and to an autoregressive (AR($p$)) process when $q = 0$.
Note that MA($q$) model is always weakly stationary and both AR($p$) and 
ARMA($p, q$) models are weakly stationary when the roots of the polynomial $\Phi(z) = 1 - \sum_{i=1}^p \varphi_i z^i$ lie outside the unit circle. We illustrate the data-generating process under the random temporal distribution shift model, where random weights follow an AR(3) model in Figure~\ref{fig:diagram}.

In the following, we give three special examples where commonly used weights for sequential datasets arise as optimal solutions under the random temporal distribution shift model. 

\paragraph{Case 1: Pooling data.} Suppose sampling variability dominates distributional uncertainty $(r \xrightarrow[]{} \infty)$. Then, the optimal weights in \eqref{eq:opt-weights} are
\begin{equation*}
    \beta^* = \argmin_{\substack{\beta:\beta^{\intercal}1 = 1  \\ \beta\geq 0}} \sum_{k=1}^{K} \beta_k^2 \frac{N/K}{n_{T+1-k}},
\end{equation*}
which lead to $\beta_k^*= n_{T+1-k} / N$. This corresponds to pooling all data.

\paragraph{Case 2: Using the most recent data.} Suppose distributional uncertainty dominates sampling uncertainty $(r \xrightarrow[]{} 0)$ and weights follow an AR(1) process. Then the optimal weighting in \eqref{eq:opt-weights} assigns most of the weight to the most recent data, with a pooling term that vanishes to zero as $K \xrightarrow[]{} \infty$, i.e., 
\begin{equation*}
\beta_k^* ={\varphi \text{I}(k=1)} + {\frac{1}{K}(1-\varphi)}.
\end{equation*}

\paragraph{Case 3: Exponential down-weighting.} Suppose distributional uncertainty dominates sampling uncertainty $(r \xrightarrow[]{} 0)$ and weights follow an ARMA(1, 1) process, i.e., 
\begin{equation*}
    W_j^t = c + \varphi W_j^{t-1} + \epsilon_j^t -\theta \epsilon_j^{t-1},
\end{equation*}
for $\varphi, \theta \in (0, 1)$ and $c \ge 0$. To ensure that the weights $W_j^t$ are non-negative, we assume that $\epsilon_j^t \le c/\theta$. Then, 
\begin{equation*}
    \epsilon^t-E[\epsilon] = W^t - \frac{1-\varphi}{1-\theta} \cdot E[W] - \sum_{k=1}^{\infty} (\varphi-\theta) \theta^{k-1} W^{t-k}.
\end{equation*}
Let $\varphi > \theta$. This gives us optimal weights: for $k \geq 1$, 
\begin{equation*}
    \beta_k^*  = \frac{1}{K}\frac{1-\varphi}{1-\theta}  + (\varphi-\theta)\theta^{k-1}.
\end{equation*}
with $K \xrightarrow[]{} \infty$. In short, each weight consists of a pooling term (which vanishes as $K \rightarrow \infty$) and an exponentially decaying term.

\section{Estimation of the optimal weights}

In this section, we present a method for estimating the optimal weights $\beta^*$. For simplicity, let sample sizes of training datasets be $n := n_t$ for all $t$. 

We consider both parametric and nonparametric approaches. Under the assumption that random weights follow a weakly stationary ARMA process as in Section~\ref{sec:time-series-modeling}, we can parametrize $\beta^*$ as a weighted sum of three components: a uniform weight, an emphasis on the most recent data, and an exponentially decaying weights. Specifically, for $k = 1, \dots, K$,
\begin{equation}\label{eq:param-form}
    \beta^*_k = \alpha_1 \cdot \frac{1}{K} + \alpha_2 \cdot I(k=1) + \alpha_3 \cdot \frac{\theta^{k-1}}{\sum_{k=1}^K \theta^{k-1}},
\end{equation}
where $\alpha_1 + \alpha_2 + \alpha_3 = 1$. The parameters $\alpha_1, \alpha_2, \alpha_3$, and $\theta$ (rate of decay) can be selected via cross-validation. 

The nonparametric approach builds on the following distributional central limit theorem (CLT), which characterizes the asymptotic behavior of weighted empirical means under random distribution shifts. 

\begin{theorem}\label{thm:clt}
    Assume the distributional perturbation model of Section~\ref{sec:model}, where random weights are weakly stationary. Suppose Assumption~\ref{ass:same-order} holds with $r = \lim m(N)/n$. 
Then, for any Borel measurable square-integrable function $\phi: \mathcal{D} \xrightarrow[]{} \mathbb{R}^d$, we have
\begin{align*}
    \sqrt{m(N)}\left(\mathbb{E}^{t}[\phi(D)]  - \sum_{k=1}^K \beta_k \hat{\mathbb{E}}^{t-k}[\phi(D)]\right) \xrightarrow[]{d}
     N\left(0, \Tilde{\delta}^2(\beta) \cdot \text{Var}_{\mathbb{P}^0}(\phi(D))\right),
\end{align*}
where for any $t > K$,  $ \Tilde{\delta}^2(\beta) = \delta^2(\beta) + r \beta^{\intercal}\beta.$
Note that $K$ is fixed, and the asymptotic behavior is considered as the size of each dataset tends to infinity. 
\end{theorem}

Theorem~\ref{thm:clt} implies that for any square-integrable function $\phi(D)$ and for any $t > K$, the covariance term of the limiting distribution is inflated by the same factor $\Tilde{\delta}^2(\beta)$. Therefore, we can estimate $\beta^*$ by minimizing this inflation factor using $L$ different test functions $\phi_1(D), \dots, \phi_{L}(D)$ as follows:
\begin{align*}
    \hat{\beta} = \argmin_{\substack{\beta: \beta^{\intercal} 1 = 1 \\ \beta \geq 0}} \frac{1}{T-K}\frac{1}{L}\sum_{t=K+1}^{T}\sum_{\ell=1}^{L}\Bigg(\hat{\mathbb{E}}^t[\phi_{\ell}(D)]  - \sum_{k=1}^K\beta_k \hat{\mathbb{E}}^{t-k}[\phi_{\ell}(D)]\Bigg)^2.
\end{align*}
Then, $\Tilde{\delta}^2(\hat{\beta})$ is consistent for $\Tilde{\delta}^2(\beta^*)$.
\begin{proposition}[Consistency of $\hat{\beta}$]\label{prop:consistency}

    Assume that the $L$ test functions $(\phi_1, \dots, \phi_L)$ are uncorrelated and have unit variances under $\mathbb{P}^0$. Let $n := n_t$ for all $t$. If either $L \xrightarrow{} \infty$, or $T \xrightarrow{} \infty$ with $\text{Cov}(W_t, W_{t-h}) \xrightarrow{} 0$ as $h \xrightarrow{} \infty$, then we have 
    \begin{equation*}
        \Tilde{\delta}^2(\hat{\beta}) \xrightarrow[]{p} \Tilde{\delta}^2(\beta^*).
    \end{equation*}
\end{proposition}

\begin{remark} We assume that the $(\phi_1, \dots, \phi_L)$ are uncorrelated and have unit variances under $\mathbb{P}^0$. If this is not the case, we can apply a linear transformation $LT$ to the test functions in a pre-processing step to obtain uncorrelated test functions with unit variances. We define the transformation matrix $LT = (\hat{\Sigma}^{\Phi})^{-\frac{1}{2}}$, where $\hat{\Sigma}^{\Phi}$ is an estimated covariance matrix $\Sigma^{\Phi}$ of $(\phi_1, \dots, \phi_L)$ on the pooled data. 
\end{remark}

When applying our method with user-defined test functions, selecting appropriate test functions remains an important question as the performance of our method depends on well-behaved test functions. In the supplementary material, Section~\ref{sec:test-functions}, we provide two examples of sets of test functions: one using only covariate information, and another incorporating outcome information. In practice, test function selection can be guided by cross-validation. In the supplementary material, Section~\ref{sec:nyc-taxi}, we compare the empirical performance of different sets of test functions.

\section{Experiments}

\subsection{Yearbook Dataset in Wild-Time \texorpdfstring{\citep{yao2023wildtime}}{[Yao et al., 2023]}}

The Wild-Time benchmark \citep{yao2023wildtime} investigates real-world temporal distribution shifts across multiple tasks including image classification \citep{ginosar2015century, christie2018fmoW}, ICU patient mortality/readmission prediction \citep{johnson2021mimic}, and text data classification \citep{misra2021sculpting, clement2019arxiv}. We apply our method as a fine-tuning step on top of a range of approaches evaluated in Wild-Time to demonstrate that our method can effectively handle residual temporal distribution shifts. Specifically, we extract last hidden-layer representations from the neural network model provided by \cite{yao2023wildtime} and train a model on top of these representations with our approach. Due to computational constraints and lack of availability of model checkpoints for other datasets \footnote{The computational constraint applies to the FMoW dataset. For the ArXiv and HuffPost datasets, model checkpoints are not publicly available. For MIMIC-IV, the provided checkpoints have a dimensional mismatch with the available data embeddings.}, we focus on the Yearbook dataset \citep{ginosar2015century}. The task is to predict gender based on yearbook images. As in \cite{yao2023wildtime}, the training set consists of data from before 1970, and the test set consists of data after 1970. 

We adopt a parametric approach to estimate the optimal weights, assuming the form in \eqref{eq:param-form}. The hyperparameters $(\alpha_1, \alpha_2, \alpha_3, \theta)$ are selected via cross-validation. Table~\ref{tab:rider_finetuning_all} reports the resulting classification accuracy, including comparisons against standard fine-tuning and a variety of baseline models: supervised learning (ERM), continual learning (EWC \citep{kirkpatrick2017overcoming}, SI \citep{zenke2017synaptic}, A-GEM \citep{chaudhry2019efficient}, Fine-tuning), invariant learning (CORAL \citep{sun2016deep}, IRM \citep{arjovsky2019invariant}, GroupDRO \citep{sagawa2020distributionally}), self-supervised learning (SimCLR \citep{chen2020simple}, SwaV \citep{caron2020unsupervised}), and ensemble learning (SWA \citep{izmailov2018averaging}). While performance is similar for GroupDRO and SwaV, our weighting strategy consistently improves classification accuracy across a wide range of methods. This demonstrates the importance of capturing residual random distribution shifts even after using methods that capture first-order effects.

\begin{table}[t]
\centering
\caption{Comparison of performance with and without fine-tuning based on our approach across models from Wild-Time \citep{yao2023wildtime}. The table presents the average accuracy (in $\%$) and standard deviation (in parentheses) across three model checkpoints for each method provided by Wild-Time, each trained on different seeds.}
\vspace{0.2em}
\begin{tabular}{lcc}
\hline
\textbf{Method} 
& \multicolumn{2}{c}{\textbf{Yearbook}} \\ 
\hline
& Without RIDER Finetuning & With RIDER Finetuning \\ 
\hline
Finetuning & 81.91 (1.21) & 82.24 (1.41)\\ 
EWC        & 77.54 (3.92) & 83.13 (1.47)\\ 
SI         & 78.70 (3.10) & 83.21 (0.50)\\ 
AGEM       & 81.03 (1.15) & 84.89 (1.16)\\
ERM        & 80.93 (3.51) & 84.03 (0.83)\\ 
CORAL      & 75.78 (1.22) & 81.29 (1.20) \\ 
GroupDRO        & 76.38 (3.71) & 76.19 (1.95) \\
IRM   & 77.98 (3.37) & 80.15 (0.67) \\ 
SimCLR     & 75.33 (1.11) & 77.07 (0.90)\\ 
SWAV       & 77.44 (5.28) & 77.61 (1.22)\\ 
SWA        & 79.81 (3.29) & 83.84 (0.78)\\ 
\hline
\end{tabular}
\label{tab:rider_finetuning_all}
\end{table}

\subsection{Predicting Realized Volatility of Stocks}

Forecasting asset return volatility is a key problem in many financial applications, such as risk management and asset allocation. 
Our goal is to predict next-day Realized Volatility (RV) estimates using data from the Oxford-Man Institute of Quantitative Finance \citep{heber2013oxford}. They provide several daily volatility estimates from 5-minute interval high-frequency data across 30 global stock indices. We focus on data from January 2005 to December 2012, a period that includes several episodes of increased volatility and market stress, such as the Global Financial Crisis from 2007 to 2008.

We simulate a scenario where the realized volatility prediction model is updated weekly. For each update, we use data from the past $K$ weeks. Treating each week as a dataset (denoted by $\mathcal{D}^w$ for a week $w$), our method estimates weights $\beta \in \mathbb{R}^K$. Since each dataset contains only 4–5 data points and just 4 covariates (as defined below), applying some of the methods from the Wild-Time benchmark in Table~\ref{tab:rider_finetuning_all} has limitations. Nonetheless, our method remains applicable on its own in this setting. Below, we outline the step-by-step procedure. 

\begin{enumerate}
    \item We collect daily data $(X_d, Y_d)$, where $d$ denotes the day, spanning from January 2005 and December 2012. The outcome $Y_d$ is the RV estimate for the following day $d+1$. The features $X_d \in \mathbb{R}^4$ include realized variance, median realized variance, negative realized semivariance, and bipower variation of the day $d$. Each week, we fit the prediction model. The target weeks range from January 2006 to December 2012. 

    As part of preprocessing, for each window of data points being considered, covariates and outcomes are clipped at the 5th and 95th percentiles. Clipping is a commonly used technique in financial time series prediction \citep{zhong2019predicting}.

    \item We define test functions simply to be covariates: $\phi_{\ell}(X_d, Y_d) = X_{d, \ell}$ for $\ell = 1, 2, 3, 4$.  

    \item For target week $w^{\text{t}}$, we estimate weights $\hat{\beta} \in \mathbb{R}^{K}$ for $K$ as
    \begin{align*}
        &\text{minimize}\sum_{w=w^{\text{t}}-K/2}^{w^{\text{t}}-1} \sum_{\ell=1}^{4} \Bigg(\hat{\mathbb{E}}^{w}[\phi_{\ell}] - \sum_{k=1}^{K}\beta_k \hat{\mathbb{E}}^{w-k}[\phi_{\ell}]\Bigg)^2, \\
        &\text{subject to } \quad \beta^{\intercal}1 = 1,  \beta \geq 0, \beta_1 \leq B, \\
        &\quad\quad\quad\quad\quad \text{  } \beta_1 \geq \beta_2 \geq \dots \geq \beta_{K}.
    \end{align*}
    We impose additional constraints including a monotonicity constraint as a form of regularization to manage a relatively large value of $K$. These weights capture how the closeness between datasets decreases as they become further apart. We also ensure that the majority of the weights are not concentrated on just the first few elements by bounding them with $B$. The upper bound $B$ is chosen based on reference exponential weights in \eqref{eq:exp-weight}, using various half-life values. The half-life value is chosen via cross-validation.

    \item We then fit the linear model. The model minimizes the following objective: 
    \begin{equation*}
       \hat{\theta}_{w^t} = \arg\min_{\theta}\textstyle\sum_{k=1}^{K} \hat{\beta}_k\sum_{d\in w^{t}-k}(Y_d- X_d^{\intercal}\theta)^2.
    \end{equation*}
    
\end{enumerate}

\begin{figure}[t]
    \centering
    \includegraphics[width=0.31\linewidth]{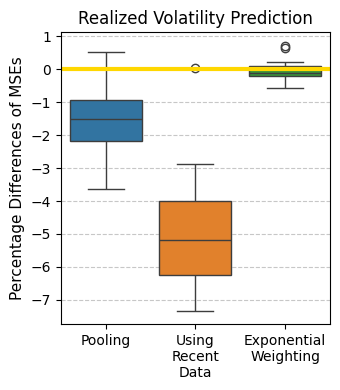} 
    \includegraphics[width=0.42\linewidth]{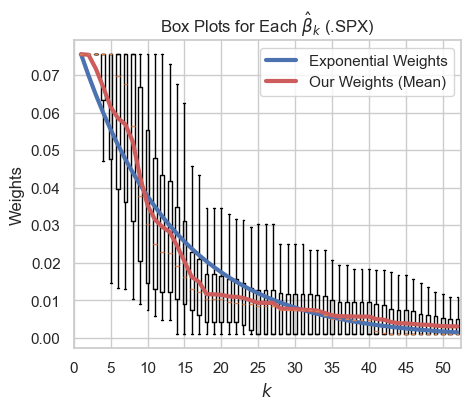}
    \caption{[Left] Boxplots showing the percentage differences in test MSEs of our method compared to three baseline methods. A negative difference means our method performs better, with $y=0$ highlighted by a yellow line. The t-tests would give significant p-values (<1e-3) for the two leftmost boxplots, while the rightmost boxplot has a p-value of 0.57. [Right] Distributions of $\hat{\beta}$ for $K=52$ across a range of target weeks from January 2006 to December 2012 for the S$\&$P 500 index.}
    \label{fig:boxplots}
\end{figure}

 We compare our method to three common practices over 15 stock market indices.\footnote{The indices cover BVSP (Brazil), CAC 40 (France), DAX (Germany), Euro Stoxx 50 (Europe), FTSE 100 (UK), SSMI (Swiss), Nasdaq (US), Russel 2000 (US), S$\&$P 500 (US), IPC (Mexico), HSI (Hong Kong), Nifty 50 (India), Nikkei 225 (Japan), SCI (China), and KOSPI (Korea).} We select the hyperparameters $[K, B]$ via cross validation. The results are summarized in Figure~\ref{fig:boxplots}. The right-hand side shows the percentage differences in test MSE (out-of-sample error) between our method and each comparison method. A negative difference indicates that our method performs better. The three methods are as follows: 1) \textbf{Pooling}: Train on the entire $K$ weeks data with equal weights, 2) \textbf{Using Recent Data}:  Train on the most recent 10 weeks data with equal weights, 3) \textbf{Exponential Weighting}: Train with exponential weights with half-life $H$ (weeks), which are commonly used in finance applications \citep{luxenberg2022exponentially}. The exponential weights are defined as
\begin{equation}\label{eq:exp-weight}
    \beta^{\text{exp}}_k = (1/2)^{k/H} / \textstyle\sum_{k=1}^{K} (1/2)^{k/H}.
\end{equation}
Note that for pooling and exponential weighting, the hyperparameters $K$ and $[K, H]$, respectively, are selected via cross-validation.

The results show that our proposed model outperforms both Method 1 and Method 2, and performs similarly to Method 3. An example of the distribution of fitted weights for $K=52$ (1 year) and $H = 9$ (63 days as in \cite{luxenberg2022exponentially}) for the S$\&$P 500 is provided in the right-hand side of Figure ~\ref{fig:boxplots}. Interestingly, our weights behave similarly to exponential weights on average, which helps explain their comparable performances. In fact, we recover exponential weighting both theoretically and empirically. This suggests that the random distribution shift model could provide a theoretical explanation for the real-world effectiveness of exponential weighting.

In the supplementary material, Section~\ref{sec:nyc-taxi}, we include an additional experiment using the NYC Taxi dataset where we explore alternative choices of test functions. By designing test functions that capture shifts in $Y|X$, we show that our proposed method can achieve some additional performance gains. 

\section{Related Work}

\textbf{Temporal Shifts.} Temporal distribution shifts have been the focus of temporal domain genaralization \citep{nasery2021training, qin2023generalizing, bai2022temporal, zeng2022latent}. Our approach, however, takes a fundamentally different direction: weighted ERM under random distribution shifts. While \cite{bennett2022time} also uses weighted ERM framework, they assume a parametric form (e.g. exponential) for ERM weights without theoretical justification and do not provide theoretical guarantees for the estimated weights.

A key distinction of our procedure is that it can be used with any empirical risk minimization procedure, with little additional computational cost. In contrast, \cite{bai2022temporal} can be computationally intensive for large-scale datasets and models, as they involve recurrent graph updates. Furthermore, unlike approaches in \citep{nassar, zeng2022latent, xie2024enhancing}, our method does not require generative models. We compute the expected excess risk under random distribution shifts, with emphasis on an equality rather than an upper bound. Prior work has primarily focused on upper bounds \citep{bartlett1992learning, barve1996complexity, long1998complexity, mohri2012new, hanneke2019statistical, mazzetto2023adaptive, zeng2024generalizing, pham2025nonstationary}.

Dynamic topic models use a state-space model to model the parameters of a sequence of parametric distributions and assume data is drawn i.i.d. from each distribution \cite{blei2006dtm, heckerman2006dtm, srebro2005time}. However, we model temporal dependencies using random weights, avoiding specific parametric assumptions about data distributions. This allows us to handle temporal distribution shifts across different data modalities. 

\textbf{Continual Learning.} Continual learning aims to effectively adapt to non-stationary distributions through sequential data streams \citep{adel2019claw, kirkpatrick2017overcoming, schwarz2018progress, zenke2017synaptic, chaudhry2019efficient, lopez2017gradient, rebuffi2017icarl}. The goal is to accumulate and reuse knowledge while balancing the trade-off between learning new information and retaining old information.

\textbf{Distributionally Robust Optimization.} 
Distributionally robust optimization (DRO) approaches aim to minimize worst-case risk over a set of potential test distributions \citep{ben2013robust, delage2010distributionally, duchi2021learning}. 
Since DRO can often be conservative in practice, the group DRO was proposed instead to optimize the worst-group loss over a collection of target distributions generated from multiple sources \citep{sagawa2020distributionally, hu2018does, wang2023distributionally}. 

\textbf{Invariant Risk Minimization and Causal Representation Learning.}
Invariant risk minimization \cite{arjovsky2019invariant} and causal representation learning \cite{scholkopf2021toward} 
aim to learn predictors that remain invariant across different environments. 
The key idea is that by focusing on stable, causal features rather than spurious correlations, models will generalize better to unseen data. 
However, empirical studies often reveal that these methods do not consistently outperform standard empirical risk minimization \cite{gulrajani2020search, koh2021wilds}.

\textbf{Importance Weighting.} 
Importance weighting techniques re-weigh the empirical loss to match the target distribution \citep{pan2010transfer, shu2019}. 
This approach has been extensively studied in both transfer learning  \citep{pan2010transfer} and robust deep learning \citep{shu2019}.
A classic application is covariate shift \citep{gretton2009covariate, shimodaira2000improving, sugiyama2008direct}. 
Recent empirical studies have found that the covariate shift assumption often does not hold in real-world datasets \citep{newlang2023hong}.

\section{Conclusion}

We have presented a framework for handling temporal distribution shifts through weighted ERM. 
Our procedure provides a principled way to interpolate between three common approaches: pooling all data, exponential down-weighting, and using only the most recent data. Under a random distribution shift model, we provide theoretical guarantees for its out-of-distribution performance. Our empirical results demonstrate that our method can serve as a complementary tool to existing approaches for addressing residual temporal distribution shifts.

This work has limitations. From a modeling perspective, real-world scenarios may involve more complex forms of distributional change than those studied here. For example, hybrid models that capture both systematic and random components of distribution shift could be more appropriate. From an application perspective, when estimating weights using user-defined test functions, developing principled approaches for selecting test functions remains an important direction for future work.

\bibliography{bibliography}

\newpage
\appendix

\section{Selection of Test Functions}\label{sec:test-functions}

The choice of test functions is crucial as they determine how we measure distributional similarity. Our method assumes two distributions are similar if the averaged values of test functions from the two datasets are closely aligned. Therefore, the choice of test functions should capture the dimensions where a close match is critical.

Examples of test functions $\phi_{\ell}(D)$ include:

\begin{enumerate}
    \item $\phi_{\ell}(D)  = X^{\ell}$ for a covariate $X^{\ell}$ so that $\hat{\mathbb{E}}^t[{\phi_{\ell}(D)}] = \hat{\mathbb{E}}^t[X^{\ell}]$. This choice is justified if there is random distribution shift in both $X$ and $Y$.
    \item $\phi_{\ell}(D_{t\bullet}) = 1_{X \in A_{\ell}} Y / \hat{\mathbb{E}}^{t}[1_{X \in A_{\ell}}]$ for a subset $A_{\ell}$ so that $\hat{\mathbb{E}}^t[{\phi_{\ell}(D)}] =\hat{\mathbb{E}}^t[Y|X \in A_{\ell}]$. This choice is justified if there is random distribution shift in $Y|X$.
\end{enumerate}

Note that $\hat{\mathbb{E}}^t[\phi_{\ell}(D)] = \frac{1}{n_t}\sum_{i=1}^{n_t}\phi_{\ell}(D_{ti})$ is the empirical mean on the $t$-th dataset. 
    In practice, it may not be clear a priori which of the two choices is more appropriate for the given dataset. We recommend selecting between the two options using cross-validation.
    
Below, we provide step-by-step implementation guides. 

\paragraph{Step-by-Step Implementation}

\begin{enumerate}
    \item Collect the datasets $\mathcal{D}_1, \dots, \mathcal{D}_T$ where $\mathcal{D}_t = \{D_{ti}=(X_{ti}, Y_{ti})\}_{i=1}^{n_t}$. The target time point is $T+1$.

    \item Choose the window size $K$ and select $L$ test functions $\phi_1(D), \dots, \phi_L(D)$. Test functions can be standardized as follows:
    \begin{equation*}
        \phi_{\ell}(D) \leftarrow \phi_{\ell}(D)/\sqrt{\widehat{\text{Var}}_{\mathbb{P}^0}(\phi_{\ell})}
    \end{equation*}
     where $\widehat{\text{Var}}_{\mathbb{P}^0}(\phi_{\ell})$ is the empirical variance of $\phi_{\ell}(D_{ti})$ on the pooled data $\{\{D_{ti}\}_{i=1}^{n_t}\}_{t=1}^T$. 

     \item Estimate weights $\hat{\beta} \in \mathbb{R}^K$ by solving $\hat{\beta}=$
    \begin{equation*}
        \argmin_{\substack{\beta: \beta^{\intercal}1 = 1,\\ \beta \geq 0} }\sum_{t=K+1}^{T}\sum_{\ell=1}^{L}\left(\hat{\mathbb{E}}^t[\phi_{\ell}(D)] - \sum_{k=1}^K\beta_k \hat{\mathbb{E}}^{t-k}[\phi_{\ell}(D)]\right)^2.
    \end{equation*}

    \item Estimate model parameters using weighted ERM:
    \begin{equation*}
        \hat{\theta}_{T+1} = \arg\min_\theta \sum_{k=1}^K \hat{\beta}_k \hat{\mathbb{E}}^{T+1-k}[\mathcal{L}(Y, f(X; \theta))].
    \end{equation*}
    Then, we predict as $\hat{Y}_{T+1} = f(X_{T+1}; \hat{\theta}_{T+1})$.
\end{enumerate}

\section{Experiments on NYC Taxi Data}\label{sec:nyc-taxi}

In this experiment, we use the publicly available NYC TLC trip record data for yellow cabs \cite{nyctlc}. Each month includes over a million daily trip records. Our goal is to predict ride durations in Manhattan and Brooklyn during the end-of-year season from December 15th to December 31st. This period is known to attract a large number of tourists in NYC, which can affect traffic and the distribution of pick-up and drop-off locations. Our goal is to model how these holiday season effects influence ride durations. 

Suppose we have data collected in the months of November of 2019, 2020, 2021, 2022 and December of 2019, 2020, 2021, with the target month being December 2022. We plan to fit a model that predicts ride durations for the target month using three datasets collected in December ($K=3$). Only the December data is used for training to model the holiday season effect. However, due to distribution shifts, the data distributions from 2019, 2020, and 2021 are likely to differ from that of 2022. To account for these potential distribution shifts across years, we apply our framework and use the November datasets to measure similarity across years. Below we outline the step-by-step procedure: 

\begin{enumerate}
    \item We collect datasets $\mathcal{D}_{yy/mm}=\{D_{yy/mm, i}\}_{i=1}^n$ for 
    $(yy, mm)$ = (19, 11), (19, 12), (20, 11), (20, 12), (21, 11), (21, 12), (22, 11),
    where $yy$ represents a year and $mm$ represents a month. The target time is (22, 12). Note that $D=(X, Y)$ where the outcome $Y\in \mathbb{R}$ represents ride duration, while the features $X\in \mathbb{R}^5$ include weekday, time of day, pick-up locations, drop-off locations, and trip distance. Each dataset consists of $n=10,000$ data points, which are random sub-samples of the original datasets. 

    \item We categorize the time of day into morning (6am-9am), midday (9am-4pm), evening (4pm-7pm), and night (7pm-6am). For each pair $p_{\ell}$
     of time of day and pick-up location, where all datasets have at least 10 points associated with this pair, let the event $A_{\ell}$ be
     \begin{equation*}
         (X \in A_{\ell}) = ((X_\text{time-of-day}, X_\text{pick-up-location}) = p_{\ell})
     \end{equation*}
    Then, for each $A_{\ell}$, we define a test function so that 
    \begin{align*}
        \hat{\mathbb{E}}^{(yy, mm)}[{\phi_{\ell}(D)}] &=\hat{\mathbb{E}}^{(yy, mm)}[g(X, Y)\text{ }|X \in A_{\ell}],
    \end{align*}
    where $g(X, Y) = X_\text{trip-distance}/Y$ represents the travel speed, measured in miles per minute. 
    These test functions are chosen to capture conditional shifts in travel speed, a function of outcomes.

    \item We use November datasets to estimate $\hat{\beta} \in \mathbb{R}^3$ as
    \begin{align*}
        \hat{\beta} = \argmin_{\substack{\beta: \beta^{\intercal}1 = 1,\\ \beta \geq 0}}& \sum_{\ell=1}^{L}\Bigg(\hat{\mathbb{E}}^{(2022, 11)}[\phi_{\ell}(D)] - \sum_{k=1}^{K=3}\beta_k \hat{\mathbb{E}}^{(2022-k, 11)}[\phi_{\ell}(D)]\Bigg)^2.
    \end{align*}
    These weights measure how similar the years 2019, 2020, and 2021 are to 2022. 

    \item We use December datasets to fit the prediction model using XGBoost in Python with the default settings. The categorical variables are recoded into dummy variables. The model minimizes the following objective: 
    \begin{equation*}
       \hat{f} = \arg\min_f\sum_{k=1}^{K=3} \hat{\beta}_k\hat{\mathbb{E}}^{(2022-k, 12)}[(Y- f(X))^2].
    \end{equation*}
    
\end{enumerate}

\begin{figure}
    \centering
    \includegraphics[width=0.31\linewidth]{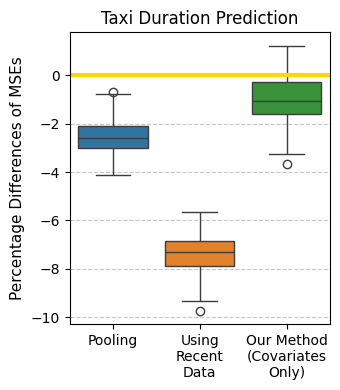}
    \includegraphics[width=0.47\linewidth]{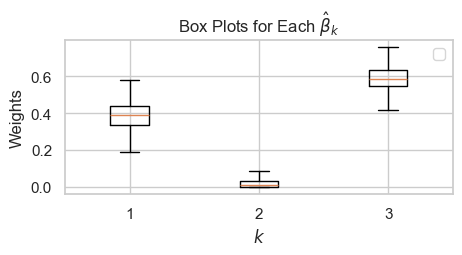}
    \caption{[Left] Boxplots showing the percentage differences in test MSEs of our method compared to three baseline methods. A negative difference means our method performs better, with $y=0$ highlighted by a yellow line. The t-tests would give significant p-values ($<$1e-5) for all the boxplots. [Right] Distributions of $\hat{\beta}$ over 100 simulations.}
    \label{fig:taxi-experiment}
\end{figure}
We compare our proposed method to three different approaches over 100 iterations, each iteration using a different sub-sample of the original datasets. The results are summarized in Figure~\ref{fig:taxi-experiment}. The left-hand side shows the percentage differences in test MSE (out-of-sample error) between our method and each comparison method. A negative difference indicates that our method performs better. The three methods are as follows: (1) \textbf{Pooling}: Train on the entire three-year datasets with equal weights for all samples, (2) \textbf{Using Recent Data}: Train on the most recent data (December 2021), (3)  
\textbf{Our Method (Covariates Only)}: Train using our proposed approach, but without outcome information in the test functions. Instead, it only uses covariates as 
 \begin{align*}
        \hat{\mathbb{E}}^{(yy, mm)}[{\phi_{\ell}(D)}] &=\hat{\mathbb{E}}^{(yy, mm)}[X_\text{trip-distance}\text{ }|X \in A_{\ell}].
    \end{align*}
The results demonstrate that our proposed model outperforms all three baselines, with statistically significant p-values. These findings highlight two key insights. First, instead of relying on standard techniques like equal weighting or using only the most recent data, our method effectively accounts for distribution shifts by assigning more appropriate weights based on past datasets. Second, the choice of test functions can be critical. By designing test functions that also capture shifts in $Y|X$, our proposed method achieves some additional performance gains compared to Method 3, where test functions just use covariates.

The distribution of fitted weights is shown in the right-hand side of Figure~\ref{fig:taxi-experiment}. Interestingly, it assigns zero weight to $\beta_2$ most of the time and slightly more weight to $\beta_3$, which corresponds to the furthest data, compared to $\beta_1$. One possible explanation for this is the impact of COVID-19, which made the year 2020 look quite different from 2022. While November and December of 2021 saw significant recovery from COVID-10, its effects were still lingering, which may explain why 2019 data received slightly higher weights than 2022 data. 

\section{Experimentation Details: Hyperparameter Selection and Model Architecture}\label{sec:exp_details}

\paragraph{Wild-Time Benchmark.} 
We use model checkpoints and the code from the official Wild-Time GitHub repository \url{https://github.com/clinicalml/wild-time}, licensed under the MIT License.

For estimating optimal weights, we conduct cross-validation over the following grid: window size $K \in \{10, 20\}$; mixture weights $\boldsymbol{\alpha} = [\alpha_1, \alpha_2, 1 - \alpha_1 - \alpha_2]$ where $\alpha_1, \alpha_2 \in \{0, 0.2, 0.4, 0.6, 0.8, 1\}$; and a decay parameter $\theta = (1/2)^{1/h}$ where $h \in \{2, 4, 6, 8\}$.

For fine-tuning, we use a 3-layer fully connected neural network. The network has hidden layer sizes of 32, 16, and 8, each followed by a ReLU activation and a dropout layer with dropout rate 0.1. The final layer is a linear output layer that maps to the number of classes. We train the model using a batch size of 128, a learning rate of 0.001, and a maximum of 200 epochs.

\paragraph{Stock Volatility Prediction.}  
For estimating optimal weights, we conduct cross-validation over the following grid: window size $K \in \{26, 39, 52, 65, 78\}$, corresponding to 6, 9, 12, 15, and 18 months. The half-life values $H$ considered are $\{6, 9, 12\}$. 

\section{Proofs}

\subsection{Proofs of Auxiliary Lemmas}

First, we restate Lemma 1 from \cite{jeong2024oodgeneralization} as an lemma~\ref{lemma:prev-paper}.
\begin{lemma}[Distributional CLT]\label{lemma:prev-paper}
Assume the distributional perturbation model mentioned in Section~\ref{sec:model}. For any bounded function $\phi_k: \mathcal{D}_k \xrightarrow[]{} \mathbb{R}$ for $k = 0, 1, \dots, K$, we have
\begin{equation}
    \sqrt{m}\left(
   \begin{pmatrix}
   {\mathbb{E}}^{t}[\phi_0(D)]\\
    {\mathbb{E}}^{t-1}[\phi_1(D)]\\
    \vdots \\
     {\mathbb{E}}^{t-K}[\phi_K(D)]
    \end{pmatrix} - 
    \begin{pmatrix}
    \mathbb{E}^0[\phi_0(D)]\\
    \mathbb{E}^0[\phi_1(D)]\\
    \vdots \\
     \mathbb{E}^0[\phi_K(D)] 
    \end{pmatrix}\right) \xrightarrow[]{d}
     N\left(0, \Sigma_t^W  \odot \text{Var}_{\mathbb{P}^0}(\phi(D))\right)
\end{equation}
where  $\phi(D)^{\intercal} = (\phi_0(D), \phi_1(D), \dots,  \phi_K(D)) \in \mathbb{R}^{K+1}$, and $\Sigma_t^W\in \mathbb{R}^{(K+1) \times (K+1)}$ has 
$$(\Sigma_t^W)_{ij} = Cov_P(W^{t-i+1}, W^{t-j+1}).$$
Here, $\odot$ denotes element-wise multiplication (Hadamard product).
\end{lemma}

In the following, we state two lemmas that would be helpful for proving Theorem~\ref{thm:clt} and Theorem~\ref{thm:ood-error}. 

\begin{lemma}\label{lemma:dist}
Assume the distributional perturbation model mentioned in Section~\ref{sec:model}. For any bounded function $\phi_k(D): \mathcal{D} \xrightarrow{} \mathbb{R}$ for $k = 1, \dots, K$, we have
\begin{equation}
    \sqrt{m}\left(
   \begin{pmatrix}
    {\mathbb{E}}^{t-1}[\phi_1(D)]\\
    \vdots \\
     {\mathbb{E}}^{t-K}[\phi_K(D)]
    \end{pmatrix} - 
    \begin{pmatrix}
    \mathbb{E}^t[\phi_1(D)]\\
    \vdots \\
     \mathbb{E}^t[\phi_K(D)] 
    \end{pmatrix}\right) \xrightarrow[]{d}
     N\left(0, \Sigma_{t, K}^W  \odot \text{Var}_{\mathbb{P}^0}(\phi(D))\right)
\end{equation}
where  $\phi(D)^{\intercal} = (\phi_1(D), \dots,  \phi_K(D)) \in \mathbb{R}^{K}$, and $\Sigma_{t, K}^W\in \mathbb{R}^{K \times K}$ has 
\begin{equation}\label{eq:sigma^w}
    (\Sigma_{t, K}^W)_{ij} = {Cov_P(W^{t-i}, W^{t-j})} + {Cov_P(W^{t}, W^{t})} - {Cov_P(W^{t-i}, W^{t})}- {Cov_P(W^{t-j}, W^{t})}.
\end{equation}
Here, $\odot$ denotes element-wise multiplication (Hadamard product).
\end{lemma}

\begin{proof}
    We apply the linear transformation to Lemma~\ref{lemma:prev-paper} on the asymptotic distribution of 
    \begin{equation*}
        \left(\mathbb{E}^t[\phi_1] - \mathbb{E}^0[\phi_1], \dots,  \mathbb{E}^t[\phi_K] - \mathbb{E}^0[\phi_K], 
         \mathbb{E}^{t-1}[\phi_1] - \mathbb{E}^0[\phi_1], \dots, 
         \mathbb{E}^{t-K}[\phi_K] - \mathbb{E}^0[\phi_K] \right) \in \mathbb{R}^{2K}
    \end{equation*}
    with a transformation matrix:
    \begin{equation*}
    A = \begin{pmatrix}
        -1 & 0 & \dots & 0 & 1 & 0 &\dots  & 0\\
         0  & -1 & \dots & 0 & 0 & 1 & \dots  & 0\\
         \vdots & \vdots &\vdots & \vdots & \vdots & \vdots & \vdots & \vdots \\
         0 & 0 & \dots & -1 &   0 & 0 & \dots  & 1\\
    \end{pmatrix}  \in \mathbb{R}^{K\times2K}.
\end{equation*}
This completes the proof. 
\end{proof}

Then, we add sampling uncertainty through the following lemma. 

\begin{lemma}\label{lemma:samp}
    Assume that the Lemma \ref{lemma:dist} holds.
    Suppose Assumption~\ref{ass:same-order}.
    Then for any bounded function $\phi_k(D): \mathcal{D} \xrightarrow{} \mathbb{R}$ for $k = 1, \dots, K$, we have
    \begin{equation}
        \sqrt{m}\left(
       \begin{pmatrix}
        \hat{\mathbb{E}}^{t-1}[\phi_1(D)]\\
        \vdots \\
         \hat{\mathbb{E}}^{t-K}[\phi_K(D)]
        \end{pmatrix} - 
        \begin{pmatrix}
        \mathbb{E}^{t}[\phi_1(D)]\\
        \vdots \\
         \mathbb{E}^{t}[\phi_K(D)] 
        \end{pmatrix}\right) \xrightarrow[]{d}
         N\left(0, (\Sigma^W_{t, K} + S) \odot \text{Var}_{\mathbb{P}^0}(\phi(D))\right)
    \end{equation}
    where $S = \text{Diag}(r_{t-1}, \dots, r_{t-K})$. Here, $\odot$ denotes element-wise multiplication (Hadamard product).
    \end{lemma}

\begin{proof}
    For any $a_1, \dots, a_K \in \mathbb{R}$, we show that 
    \begin{equation*}
        \sqrt{m}\sum_{k=1}^K a_k \left( \hat{\mathbb{E}}^{t-k}[\phi_k(D)] -  \mathbb{E}^{t}[\phi_k(D)]\right)  \xrightarrow[]{d}  N\left(0, \vec{a}^{\intercal}\left( (\Sigma^W_{t, K} + S) \odot \text{Var}_{\mathbb{P}^0}(\phi(D))\right)\vec{a}\right).
    \end{equation*}
    Then, by the Cramer-Wold theorem, we have the desired result. 
    We want to show that for any $x \in \mathbb{R}$, 
    \begin{align}\label{eq:cdf}
        &E_{\xi}\left[P\left( \sqrt{m}\sum_{k=1}^K a_k \left( \hat{\mathbb{E}}^{t-k}[\phi_k(D)] -  \mathbb{E}^{t}[\phi_k(D)]\right) \leq x \cdot \sqrt{\vec{a}^{\intercal}((\Sigma^W_{t, K} + S) \odot \text{Var}_{\mathbb{P}^0}(\phi(D)))\vec{a}}\text{  }\Big | \text{ }\xi\right) \right] \nonumber \\ &= \Phi(x) + o(1), 
    \end{align}
    where $\Phi$ is the cdf of a standard Gaussian random variable and $\xi = \{W_j^{t-k}\}_{j=1, \dots, m,\text{ } k = 0, 1, \dots, K}$ encodes random perturbations. Note that the inner conditional expectation handles sampling variability (fixing random perturbations) and the outer expectation takes expectation over random perturbations. 

    Let us define 
    \begin{align*}
        A_n & =x \cdot \sqrt{\vec{a}^{\intercal}((\Sigma^W_{t, K} + S)\odot \text{Var}_{\mathbb{P}^0}(\phi(D))) \vec{a} }- \sqrt{m}\sum_{k=1}^K a_k\left( \mathbb{E}^{t-k}[\phi_k(D)] -  \mathbb{E}^{t}[\phi_k(D)]\right),\\
        B_n &= \sum_{k=1}^K a_k^2 \frac{\text{Var}_{\mathbb{P}^{t-k}}(\phi_k)}{n_{t-k}}.
    \end{align*}
    Then, we have
    \begin{equation*}
        \eqref{eq:cdf} = E_{\xi}\left[P\left( \sum_{k=1}^K\sum_{i=1}^{n_{t-k}} \frac{a_k}{n_{t-k}} \left( \phi_k(D_{t-k,i}) -  \mathbb{E}^{t-k}[\phi_k(D)]\right) \leq \sqrt{1/m} \cdot A_n \text{  }\Big | \text{ }\xi\right) \right].
    \end{equation*}
    By Berry-Essen (for non-identically distributed summands), we have that for all $n_{t-1}, \dots, n_{t-K}$, 
    \begin{equation*}
        \sup_y \Bigg|P\left( \sum_{k=1}^K\sum_{i=1}^{n_{t-k}} \frac{a_k}{n_{t-k}} \left( \phi_k(D_{t-k, i}) -  \mathbb{E}^{t-k}[\phi_k(D)]\right) \leq y \cdot \sqrt{B_n}\text{  }\Big | \text{ }\xi\right) -\Phi(y) \Bigg| \leq C_0 \cdot \psi_0
    \end{equation*}
    where 
    \begin{equation*}
        \psi_0 = \frac{\sum_{k=1}^K a_k^3\mathbb{E}^{t-k}[|\phi_k - \mathbb{E}^{t-k}[\phi_k] |^3] (N/n_{t-k})^2}{\left(\sum_{k=1}^K a_k^2\mathbb{E}^{t-k}[|\phi_k - \mathbb{E}^{t-k}[\phi_k] |^2] (N/n_{t-k})\right)^{3/2}\cdot \sqrt{N}}. 
    \end{equation*}
From the boundness assumption and Lemma~\ref{lemma:dist}, we have that 
\begin{align*}
    &\frac{\sum_{k=1}^K a_k^3\mathbb{E}^{t-k}[|\phi_k - \mathbb{E}^{t-k}[\phi_k] |^3] (N/n_{t-k})^2}{\left(\sum_{k=1}^K a_k^2\mathbb{E}^{t-k}[|\phi_k - \mathbb{E}^{t-k}[\phi_k] |^2] (N/n_{t-k})\right)^{3/2}} \\ &\xrightarrow[]{p}\frac{\sum_{k=1}^K a_k^3\mathbb{E}^{0}[|\phi_{t-k} - \mathbb{E}^{0}[\phi_k] |^3] (1/\pi_{t-k})^2}{\left(\sum_{k=1}^K a_k^2\mathbb{E}^{0}[|\phi_{t-k} - \mathbb{E}^{0}[\phi_k] |^2] (1/\pi_{t-k})\right)^{3/2}} < \infty
\end{align*}
as $N \xrightarrow{} \infty$, where $\frac{1}{\pi_{t-k}} = r_{t-k} \sum_{k=1}^K\frac{1}{r_{t-k}}$. 
Then, we have
 \begin{equation*}
        \sup_y \Bigg|P\left( \sum_{k=1}^K\sum_{i=1}^{n_{t-k}} \frac{a_k}{n_{t-k}} \left( \phi_k(D_{t-k,i}) -  \mathbb{E}^{t-k}[\phi_k(D)]\right) \leq y \cdot \sqrt{B_n}\text{  }\Big | \text{ }\xi\right) -\Phi(y) \Bigg| \xrightarrow[]{p} 0
    \end{equation*}
    Define $g_{n}(y; \xi)$ as
\begin{equation*}
    g_{n}(y ; \xi) = P\left( \frac{1}{\sqrt{B_n}}\sum_{k=1}^K\sum_{i=1}^{n_{t-k}} \frac{a_k}{n_{t-k}} \left( \phi_k(D_{t-k,i}) -  \mathbb{E}^{t-k}[\phi_k(D)]\right) \leq y \text{  }\Big | \text{ }\xi\right).
\end{equation*}
Note that
\begin{align*}
    \eqref{eq:cdf} &= E_{\xi}\left[g_n\left( A_n/\sqrt{m \cdot B_n}\right)\right] \\ &= E_{\xi}\left[g_n\left( A_n/\sqrt{m \cdot B_n}\right)\right] - E_{\xi}\left[\Phi\left( A_n/\sqrt{m \cdot B_n}\right)\right] + E_{\xi}\left[\Phi\left( A_n/\sqrt{m \cdot B_n}\right)\right]. \\ 
\end{align*}
By the dominated convergence theorem, we have
\begin{align*}
    \Big| \eqref{eq:cdf} - E_{\xi}\left[\Phi\left( A_n/\sqrt{m \cdot B_n}\right)\right]\Big|
    &\leq \Big| E_{\xi}\left[g_n\left( A_n/\sqrt{m \cdot B_n}\right)\right] - E_{\xi}\left[\Phi\left( A_n/\sqrt{m \cdot B_n}\right)\right]\Big|
    \\
    &\leq 
     E_{\xi}\left[\sup_y|g_{n}(y) - \Phi(y)|\right]= o(1).
\end{align*}
From Lemma~\ref{lemma:dist}, we have
\begin{align*}
    A_n  &\xrightarrow[]{d} x \cdot \sqrt{\vec{a}^{\intercal}((\Sigma^W_{t, K} + S)\odot \text{Var}_{\mathbb{P}^0}(\phi(D)))\vec{a}} -  \sqrt{\vec{a}^{\intercal}(\Sigma^W_{t, K} \odot \text{Var}_{\mathbb{P}^0}(\phi(D))) \vec{a}}\cdot Z, \\
    \sqrt{m\cdot B_n}  &\xrightarrow[]{p} \sqrt{\vec{a}^{\intercal}(S\odot \text{Var}_{\mathbb{P}^0}(\phi(D)))\vec{a}}
\end{align*}
where $Z$ is a standard Gaussian random variable. Let 
$$
\delta^2 = \frac{\vec{a}^{\intercal}((\Sigma^W_{t, K} + S) \odot \text{Var}_{\mathbb{P}^0}(\phi(D)))\vec{a}}{\vec{a}^{\intercal}(S\odot \text{Var}_{\mathbb{P}^0}(\phi(D)))\vec{a}}.
$$
Combining results, we have 
\begin{equation*}
    A_n/\sqrt{m \cdot B_n} \xrightarrow[]{d} \delta x - \sqrt{\delta^2-1} Z
\end{equation*}
where $Z$ is a standard Gaussian random variable.
Since $\Phi$ is bounded and continuous, by Portmanteau Lemma, we get
\begin{equation*}
    \lim_{N \xrightarrow[]{} \infty} E_{\xi}[\Phi( A_n/\sqrt{m \cdot B_n} )] = E[\Phi(\delta x - \sqrt{\delta^2-1} Z)] = \Phi(x).
\end{equation*}
This completes the proof.
\end{proof}

\subsection{Proof of Theorem~\ref{thm:clt}}

\begin{proof}
    From Lemma~\ref{lemma:dist} and Lemma~\ref{lemma:samp}, we have that for any bounded function $\phi(D): \mathcal{D} \xrightarrow{} \mathbb{R}^d$, 
    \begin{align*}
    \sqrt{m}\left(\mathbb{E}^{t}[\phi(D)]  - \sum_{k=1}^K \beta_k \hat{\mathbb{E}}^{t-k}[\phi(D)]\right) \xrightarrow[]{d}
     N\left(0, \Tilde{\delta}^2 \cdot \text{Var}_{\mathbb{P}^0}(\phi(D))\right)
\end{align*}
where $\Tilde{\delta}^2= \beta^{\intercal}(\Sigma^W_{t, K}+ S) \beta$ with 
\begin{align*}
    (\Sigma_{t, K}^W)_{ij} &= {Cov_P(W^{t-i}, W^{t-j})} + {Cov_P(W^{t}, W^{t})} - {Cov_P(W^{t-i}, W^{t})}- {Cov_P(W^{t-j}, W^{t})}, \\
    S &= \text{Diag}(r_{t-1}, \dots, r_{t-K}).
\end{align*}
The extension to multivariate function $\phi(D) = (\phi_1(D), \dots, \phi_d(D))^{\intercal} \in \mathbb{R}^d$ comes from considering a sequence:
\begin{equation*}
    (\hat{\mathbb{E}}^{t-1}[\phi_1(D)],\dots, \hat{\mathbb{E}}^{t-1}[\phi_d(D)], 
    \dots \dots, \hat{\mathbb{E}}^{t-K}[\phi_1(D)], \dots, \hat{\mathbb{E}}^{t-K}[\phi_d(D)]),
\end{equation*}
in Lemma~\ref{lemma:dist} and Lemma~\ref{lemma:samp}.

Note that $\Tilde{\delta}^2$ can be written as 
\begin{equation*}
    \Tilde{\delta}^2 =\left(E[(W^t - \sum_{k=1}^K \beta_k W^{t-k})^2] + \sum_{k=1}^K \beta_k^2 r_{t-k}\right).
\end{equation*}

We now extend this result from bounded functions $\phi $ to square integrable functions $\phi \in L^2(\mathbb{P})$. Without loss of generality for notational simplicity we restrict ourselves to the case $d=1$. Note that for any $\phi \in L^2(\mathbb{P})$ and for any given $\epsilon >0$, there exists a bounded function $\phi^B$ such that $\mathbb{E}^0[\phi(D)] = \mathbb{E}^0[\phi^B(D)]$ and $||\phi - \phi^B||_{L^2(\mathbb{P}^0)} < \epsilon$.

Applying Chebychev conditionally on random perturbations,
\begin{equation*}
    P_{\xi}\left( \Bigg| \sqrt{m} \sum_{k=1}^K\beta_k \hat{\mathbb{E}}^{t-k} [\phi(D) - \phi^B(D)] \Bigg| \geq \epsilon \text{ } \Bigg| \text{ } \xi\right) \leq \frac{m}{\epsilon^2} \sum_{k=1}^{K}\frac{\beta_k^2}{n_{t-k}}\text{Var}_{\mathbb{P}^{t-k}}(\phi^B - \phi)
\end{equation*}
where $\xi = \{W_j^{t-k}\}_{j=1, \dots, m,\text{ } k =0, 1, \dots, K}$ encodes random perturbations. Note that the conditional expectation on random perturbations handles sampling variability. Take expectation over random perturbations and we get
\begin{align*}
    P\left(  \sqrt{m} \sum_{k=1}^K \beta_k \hat{\mathbb{E}}^{t-k} [\phi(D) - \phi^B(D)] \Bigg| \geq \epsilon \right) &\leq \frac{m}{\epsilon^2} \sum_{k=1}^{K}\frac{\beta_k^2}{n_{t-k}}\mathbb{E}^0[(\phi - \phi^B)^2] \\
    &= \frac{1}{\epsilon^2} \sum_{k=1}^{K}\beta_k^2r_{t-k}\mathbb{E}^0[(\phi - \phi^B)^2] +o(1).
\end{align*}
Note that with $D = h(U)$ and $\psi = (\phi - \phi^B)\circ h$,
    \begin{align*}
        \sqrt{m}(\mathbb{E}^{t}[\phi(D)] - \mathbb{E}^{t}[\phi^B(D)]) &= \frac{\sqrt{m}\sum_{j=1}^{m}\int_{x \in I_j} \psi(x) dx \cdot (W_j^t - E[W^t])}{\sum_{j=1}^{m}W_j^t/m}.
    \end{align*}
    The variance of the numerator is bounded as
    \begin{align*}
        \text{Var}(W^t) \sum_{j=1}^{m}m\left(\int_{x \in I_j}\psi(x)dx\right)^2 
        & \leq \text{Var}(W^t) \sum_{j=1}^{m}\int_{x \in I_j}\psi^2(x)dx \\
        & = \text{Var}(W^t) \mathbb{E}^0[\psi^2(U)] = \text{Var}(W^t) \mathbb{E}^0[(\phi- \phi^B)^2]
    \end{align*}
    where the first inequality holds by Jensen's inequality with $m \int_{x\in I_j} dx = 1$.  The denominator converges in probability to $E[W^t] = 1$. Therefore, we can write
    \begin{equation*}
         \sqrt{m}(\mathbb{E}^{t}[\phi(D)] - \mathbb{E}^{t}[\phi^B(D)]) = \sqrt{m}\sum_{j=1}^{m}\int_{x \in I_j} \psi(x) dx \cdot (W_j^t - E[W^t]) + R_n^t
    \end{equation*}
    where $R_n^t$ is $o_p(1)$. By Chebychev's inequality, we have
    \begin{equation*}
         P\left(  \Big|\sqrt{m}(\mathbb{E}^{t}[\phi(D) - \phi^B(D)]) - R_n^t \Big| \geq \epsilon \right) \leq \frac{1}{\epsilon^2} \text{Var}(W^t) \mathbb{E}^0[(\phi- \phi^B)^2]
    \end{equation*}
Combining results, for any $\epsilon, \epsilon' > 0$, we can find a bounded function $\phi^B$ as $m(N) \xrightarrow[]{} \infty$ such that
\begin{align*}
    P\left(\sqrt{m} \Bigg|\mathbb{E}^{t}[\phi- \phi^B] - \sum_{k=1}^K \beta_k \hat{\mathbb{E}}^{t-k}[\phi-\phi^B]\Bigg| \geq \epsilon\right) &\leq P\left(\sqrt{m} \Big|\mathbb{E}^{t}[\phi- \phi^B] - R_n^t\Big| \geq \frac{\epsilon}{3}\right)\\ &+ P\left(|R_n^t| \geq \frac{\epsilon}{3}\right) \\
    &+ P\left(\sqrt{m} \Big|\sum_{k=1}^K \beta_k \hat{\mathbb{E}}^{t-k}[\phi-\phi^B]\Big| \geq \frac{\epsilon}{3}\right)  \leq \epsilon'.
\end{align*}
We already have convergence results for bounded functions: 
\begin{equation*}
     \sqrt{m}\left(\mathbb{E}^{t}[\phi^B(D)]  - \sum_{k=1}^K \beta_k \hat{\mathbb{E}}^{t-k}[\phi^B(D)]\right) \xrightarrow[]{d}
     N\left(0, \Tilde{\delta}^2 \cdot \text{Var}_{\mathbb{P}^0}(\phi^B(D))\right).
\end{equation*}
Finally, for any $\epsilon, \epsilon', \epsilon'' > 0$, we can find a bounded function $\phi^B$ as $m(N) \xrightarrow[]{} \infty$ such that
\begin{align*}
    &P \left( \frac{\sqrt{m}(\mathbb{E}^t[\phi(D)] - \sum_{k=1}^K \beta_k \hat{\mathbb{E}}^{t-k}[\phi(D)])}{\Tilde{\delta}\sqrt{\text{Var}_{\mathbb{P}^0}(\phi(D))}}\leq x \right) \\
    & = P \left( \frac{\sqrt{m}(\mathbb{E}^t[\phi^B] - \sum_{k=1}^K \beta_k \hat{\mathbb{E}}^{t-k}[\phi^B])}{\Tilde{\delta}\sqrt{\text{Var}_{\mathbb{P}^0}(\phi^B)}} \cdot \frac{\sqrt{\text{Var}_{\mathbb{P}^0}(\phi^B)}}{\sqrt{\text{Var}_{\mathbb{P}^0}(\phi)}} + \frac{\sqrt{m}(\mathbb{E}^t[\phi-\phi^B] - \sum_{k=1}^K \beta_k \hat{\mathbb{E}}^{t-k}[\phi-\phi^B])}{\Tilde{\delta}\sqrt{\text{Var}_{\mathbb{P}^0}(\phi)}}\leq x \right) \\
    & \leq P \left( \frac{\sqrt{m}(\mathbb{E}^t[\phi^B] - \sum_{k=1}^K \beta_k \hat{\mathbb{E}}^{t-k}[\phi^B])}{\Tilde{\delta}\sqrt{\text{Var}_{\mathbb{P}^0}(\phi^B)}} \cdot \frac{\sqrt{\text{Var}_{\mathbb{P}^0}(\phi^B)}}{\sqrt{\text{Var}_{\mathbb{P}^0}(\phi)}} \leq x + \epsilon \right) \\
    &+P \left( \frac{\sqrt{m}|\mathbb{E}^t[\phi-\phi^B] - \sum_{k=1}^K \beta_k \hat{\mathbb{E}}^{t-k}[\phi-\phi^B]|}{\Tilde{\delta}\sqrt{\text{Var}_{\mathbb{P}^0}(\phi)}} > \epsilon \right) \\
    &\leq \Phi\left((x+\epsilon) \cdot \frac{\sqrt{\text{Var}_{\mathbb{P}^0}(\phi)}}{\sqrt{\text{Var}_{\mathbb{P}^0}(\phi^B)}}\right) + \epsilon',
\end{align*}
and $|\sqrt{\text{Var}_{\mathbb{P}^0}(\phi^B)}/\sqrt{\text{Var}_{\mathbb{P}^0}(\phi)} - 1 | \leq \epsilon''$. Therefore, 
\begin{equation*}
    \limsup_{N \xrightarrow[]{}\infty}P \left( \frac{\sqrt{m}(\mathbb{E}^t[\phi(D)] - \sum_{k=1}^K \beta_k \hat{\mathbb{E}}^{t-k}[\phi(D)])}{\Tilde{\delta}\sqrt{\text{Var}_{\mathbb{P}^0}(\phi(D))}}\leq x \right) \leq \Phi(x). 
\end{equation*}
With an analogous argument, 
\begin{equation*}
    \liminf_{N \xrightarrow[]{}\infty}P \left( \frac{\sqrt{m}(\mathbb{E}^t[\phi(D)] - \sum_{k=1}^K \beta_k \hat{\mathbb{E}}^{t-k}[\phi(D)])}{\Tilde{\delta}\sqrt{\text{Var}_{\mathbb{P}^0}(\phi(D))}}\leq x \right) \geq \Phi(x). 
\end{equation*}
This completes the proof for one-dimensional $\phi$. The result for a vector of functions $\phi$ follows by applying the Cram\'er--Wold device.
\end{proof}

\subsection{Proof of Theorem~\ref{thm:ood-error}}\label{sec:proof-of-ood-error}

\begin{proof}
The first part of proof follows \cite{van2000asymptotic}, Theorem 5.41 with 
\begin{align*}
    \Psi_n(\theta) &= \sum_{k=1}^K \beta_k \hat{\mathbb{E}}^{t-k}[\partial_\theta \mathcal{L}(Y, f(X; \theta))] \\
    \dot{\Psi}_n(\theta) &= \sum_{k=1}^K \beta_k \hat{\mathbb{E}}^{t-k}[ \partial_\theta^2 \mathcal{L}(Y, f(X; \theta))].
\end{align*}
By Taylor's theorem there exists a (random vector) $\tilde{\theta}$ on the line segment between $\theta^0$ and $\hat{\theta}^{\beta}_t$ such that 
\begin{equation*}
    0 = \Psi_n(\hat{\theta}_t^{\beta}) = \Psi_n(\theta^0) +  \dot{\Psi}_n(\theta^0) (\hat{\theta}^{\beta}_t - \theta^0)+ \frac{1}{2}(\hat{\theta}^{\beta}_t - \theta^0)^{\intercal} \Ddot{\Psi}_n(\tilde \theta) (\hat{\theta}^{\beta}_t - \theta^0).
\end{equation*}
By Lemma~\ref{thm:clt} with $\phi(D) = \partial_\theta^2 \mathcal{L}(\theta^0,D)$, we have
\begin{align}
    \dot{\Psi}_n(\theta^0)  &= \sum_{k=1}^K\beta_k\frac{1}{n_{t-k}} \sum_{i=1}^{n_{t-k}} \partial_\theta^2 \mathcal{L}(Y_{t-k,i}, f(X_{t-k,i}; \theta^0)) \nonumber \\ &= \sum_{k=1}^K \beta_k \mathbb{E}^0[ \partial_\theta^2 \mathcal{L}(Y, f(X; \theta^0)) ] + o_P(1) = \mathbb{E}^0[ \partial_\theta^2 \mathcal{L}(Y, f(X; \theta^0)) ] + o_P(1). \label{eq:lemma1-pf-1}
\end{align}
By assumption, there exists a ball $B$ around $\theta^0$ such that $\theta \xrightarrow[]{} \partial_\theta^3 \mathcal{L}(\theta,D)$ is dominated by a fixed function $h(\cdot)$ for every $\theta \in B$. The probability of the event $\{\hat{\theta}^{\beta}_t \in B\}$ tends to 1. On this event,
\begin{equation*}
    \|\Ddot{\Psi}_n(\tilde \theta)\| = \Big|\Big| \sum_{k=1}^K\beta_k\frac{1}{n_k} \sum_{i=1}^{n_{t-k}} \partial_\theta^3 \mathcal{L}(Y_{t-k,i}, f(X_{t-k,i}; \tilde \theta))\Big|\Big| \leq \sum_{k=1}^K \beta_k\frac{1}{n_{t-k}} \sum_{i=1}^{n_{t-k}} h(D_{t-k,i}).
\end{equation*}
Using Theorem~\ref{thm:clt} with $\phi(D) = h(D)$, we get
\begin{equation}\label{eq:lemma1-pf-2}
       \| \Ddot{\Psi}_n(\tilde \theta) \|   \le \sum_{k=1}^K\beta_k\frac{1}{n_{t-k}} \sum_{i=1}^{n_{t-k}} h(D_{t-k,i}) = O_P(1).
\end{equation}
Combining \eqref{eq:lemma1-pf-1} and \eqref{eq:lemma1-pf-2} gives us
\begin{align*}
    -\Psi_n(\theta^0) & = \left(\mathbb{E}^0[ \partial_\theta^2 \mathcal{L}(Y, f(X; \theta^0)) ] + o_P(1) + \frac{1}{2}(\hat{\theta}_t^{\beta} - \theta^0)\text{ }O_P(1)\right)(\hat{\theta}_t^{\beta} - \theta^0) \\ &= \left(\mathbb{E}^0[ \partial_\theta^2 \mathcal{L}(Y, f(X; \theta^0)) ] + o_P(1) \right)(\hat{\theta}_t^{\beta} - \theta^0).
\end{align*}
The probability that the matrix $\mathbb{E}^0[ \partial_\theta^2 \mathcal{L}(\theta^0,D) ] + o_P(1) $ is invertible tends to 1.  Multiplying the preceding equation by $\sqrt{m}$ and applying $(\mathbb{E}^0[ \partial_\theta^2 \mathcal{L}(\theta^0,D) ] + o_P(1))^{-1} $ left and right gives us that
\begin{equation*}
    \sqrt{m}\left(\hat{\theta}_t^{\beta} - \theta^0\right) = \sqrt{m}\left(\sum_{k=1}^K \beta_k \hat{\mathbb{E}}^k[\phi(D)] - \mathbb{E}^0[\phi(D)]\right) + o_P(1),
\end{equation*}
where $\phi(D) = - \mathbb{E}^0[\partial_\theta^2 \mathcal{L}(\theta^0,D)]^{-1}  \partial_\theta \mathcal{L}(\theta^0,D)$. Then, we apply the same procedure with 
\begin{align*}
    \Psi_n(\theta) &=  \mathbb{E}^{t}[\partial_\theta \mathcal{L}(Y, f(X; \theta))] \\
    \dot{\Psi}_n(\theta) &=  \mathbb{E}^{t}[\partial_\theta^2 \mathcal{L}(Y, f(X; \theta))].
\end{align*}
Analogously, we get
\begin{equation*}
    \sqrt{m}\left(\theta_t - \theta^0\right) = \sqrt{m}\left( {\mathbb{E}}^t[\phi(D)] - \mathbb{E}^0[\phi(D)]\right) + o_P(1).
\end{equation*}
Combining results, we get
\begin{equation}\label{eq:theta_asymp_lin}
    \sqrt{m}\left(\hat{\theta}_t^{\beta} - \theta_t\right) = \sqrt{m}\left(\sum_{k=1}^K \beta_k \hat{\mathbb{E}}^k[\phi(D)] - \mathbb{E}^t[\phi(D)]\right) + o_P(1).
\end{equation}
By Theorem~\ref{thm:clt}, we have
\begin{equation*}
    \sqrt{m}\left(\hat{\theta}_t^{\beta} - \theta_t\right) \xrightarrow[]{d} N(0, \Tilde{\delta}^2 \cdot  \text{Var}_{\mathbb{P}^0}(\phi(D))),
\end{equation*}
where
\begin{equation*}
    \Tilde{\delta}^2 = E[(W^{t} - \sum_{k=1}^K \beta_k W^{t-k})^2] + \sum_{k=1}^K \beta_k^2 r_{t-k}.
\end{equation*}

Now let $\Phi(\theta) =  \mathbb{E}^t[\partial_{\theta}\mathcal{L}(Y, f(X;\theta))], \dot{\Phi}(\theta) = \mathbb{E}^t[\partial_{\theta}^2\mathcal{L}(Y, f(X; \theta))], \Ddot{\Phi}(\theta) = \mathbb{E}^t[\partial_{\theta}^3\mathcal{L}(Y, f(X; \theta))]$. By Taylor's theorem, there exist $\tilde{\theta}, \tilde{\tilde{\theta}}$,  on the line segment between $\hat{\theta}^{\beta}_t$ and $\theta_t$ such that 
\begin{align*}
     \mathbb{E}^t[\mathcal{L}(Y, f(X; \hat \theta^\beta_t))]  -  \mathbb{E}^t[\mathcal{L}( Y, f(X; \theta_t))] &= (\hat{\theta}^{\beta}_t - \theta_t)^{\intercal}{\Phi}(\theta_t)  + \frac{1}{2}(\hat{\theta}^{\beta}_t - \theta_t)^{\intercal} \dot{\Phi}(\tilde \theta) (\hat{\theta}^{\beta}_t - \theta_t) \\
     &= \frac{1}{2}(\hat{\theta}^{\beta}_t - \theta_t)^{\intercal} \left( \dot{\Phi}(\theta_t) + \Ddot{\Phi}(\tilde{\tilde{\theta}}) (\tilde{\theta} - \theta_t)\right) (\hat{\theta}^{\beta}_t - \theta_t) \\
     & = \frac{1}{2}(\hat{\theta}^{\beta}_t - \theta_t)^{\intercal} \left( \dot{\Phi}(\theta_t) + o_P(1) \right) (\hat{\theta}^{\beta}_t - \theta_t),
\end{align*}
as ${\Phi}(\theta_t) = 0$ and $\Ddot{\Phi}(\tilde{\tilde{\theta}}) \leq \mathbb{E}^0[h(D)] + o_P(1)$ with probability approaching 1. Similarly, there exists $\theta'$ on the line segment between $\theta_t$ and $\theta^0$ such that 
\begin{equation*}
    \dot{\Phi}(\theta_t) = 
    \dot{\Phi}(\theta^0) +
    \Ddot{\Phi}(\theta')(\theta' - \theta^0)  = \mathbb{E}^0[\partial_{\theta}^2\mathcal{L}(Y, f(X; \theta^0))] + o_P(1).
\end{equation*}
Then  the rescaled excess risk can be written as 
\begin{align}\label{eq:risk-taylor}
  m \cdot \left( \mathbb{E}^t[\mathcal{L}(Y, f(X; \hat \theta^\beta_t))]  -  \mathbb{E}^t[\mathcal{L}(Y, f(X;  \theta_t))] \right)\nonumber \\ = m \cdot (\hat \theta^\beta_t - \theta_t)^\intercal \mathbb{E}^0[\partial_\theta^2 \mathcal{L}(Y, f(X;\theta^0))]   (\hat \theta^\beta_t - \theta_t) + o_P(1).
\end{align}
Using \eqref{eq:theta_asymp_lin}, asymptotically \eqref{eq:risk-taylor} follows a distribution with asymptotic mean
\begin{equation*}
 \Tilde{\delta}^2 \cdot \mathrm{Trace}( \mathbb{E}^0[\partial_\theta^2 \mathcal{L}(Y, f(X; \theta^0))]^{-1} \mathrm{Var}_{\mathbb{P}^0}(  \partial_\theta \mathcal{L}(Y, f(X; \theta^0)) ))
\end{equation*}
This completes the proof. 
\end{proof}

\subsection{Consistency}~\label{sec:consistency}
In the following, we add regularity conditions that lead to the consistency of M-estimators, $\hat{\theta}_{t}^{\beta}$. 

\begin{lemma}[Consistency of M-estimators]\label{lemma:consistency} Consider the M-estimator
$$
\hat{\theta}_t^{\beta} = \arg \min_{\theta \in \Omega} \sum_{k=1}^K \beta_k \hat{\mathbb{E}}^{t-k}[L(\theta, D)] 
$$
and the target $\theta^0 = \arg \min_{\theta\in \Omega}\mathbb{E}^0[L(\theta, D)]$, where $\Omega$ is a compact subset of $\mathbb{R}^d$. 
    Assume that $\theta \xrightarrow[]{} L(\theta, D)$ is continuous and that $\inf_{||\theta - \theta'||_2 \leq \delta}L(\theta, D)$ is square-integrable under $\mathbb{P}^0$ for every $\delta$ and $\theta'$ and that $\inf_{\theta \in \Omega}L(\theta, D)$ is square-integrable. We assume that $\mathbb{E}^0[L(\theta, D)]$ has a unique minimum. Then, $\hat{\theta}_t^{\beta} - \theta^0 = o_P(1), \theta_t - \theta^0 = o_p(1)$.
\end{lemma}

\begin{proof}
    The proof follows \cite{van2000asymptotic}, Theorem 5.14 with $m_\theta(D) = - L(\theta,D)$. 
    
    Fix some $\theta$ and let $U_{\ell} \downarrow \theta$ be a decreasing sequence of open balls around $\theta$ of diameter converging to zero. Write $m_{U}(D)$ for $\sup_{\theta \in U}m_{\theta}(D)$. The sequence $m_{U_{\ell}}$ is decreasing and greater than $m_{\theta}$ for every $\ell$. As $\theta \xrightarrow[]{} m_{\theta}(D)$ is continuous, by monotone convergence theorem, we have $\mathbb{E}^0[m_{U_{\ell}}(D)] \downarrow \mathbb{E}^0[m_{\theta}(D)]$. 

    For $\theta \neq \theta^0$, we have $\mathbb{E}^0[m_\theta(D)] < \mathbb{E}^0[m_{\theta^0}(D)]$.  Combine this with the preceding paragraph to see that for every $\theta \neq \theta^0$ there exits an open ball $U_{\theta}$ around $\theta$ with $\mathbb{E}^0[m_{U_{\theta}}(D)] < \mathbb{E}^0[m_{\theta^0}(D)]$. For any given $\epsilon > 0$, let the set $B = \{\theta \in \Omega : ||\theta - \theta^0||_2 \geq \epsilon\}$. The set $B$ is compact and is covered by the balls $\{U_{\theta}: \theta \in B\}$. Let $U_{\theta_1},\dots, U_{\theta_p}$ be a finite sub-cover. Then, 
\begin{align}
    \sup_{\theta \in B} \sum_{k=1}^K\beta_k \frac{1}{n_k} \sum_{i=1}^{n_k} m_{\theta}(D_{t-k,i}) & \leq \sup_{j = 1, \dots, p} \sum_{k=1}^K\beta_k \frac{1}{n_{t-k}} \sum_{i=1}^{n_{t-k}} m_{U_{\theta_j}}(D_{t-k,i})
    \nonumber \\ &= \sup_{j = 1, \dots, p} \mathbb{E}^0[ m_{U_{\theta_j}}(D)] + o_P(1) < \mathbb{E}^0[m_{\theta^0}(D)] + o_P(1) \label{eq:pf-lemma2},
\end{align}
where for the equality we apply Theorem~\ref{thm:clt} (without the linear transformation in Lemma~\ref{lemma:dist}) with $\phi(D) = m_{U_{\theta_j}}(D)$ for all $j=1,\ldots,p$. 

If $\hat{\theta}_t^{\beta} \in B$, then 
\begin{align*}
    \sup_{\theta \in B} \sum_{k=1}^K\beta_k \frac{1}{n_{t-k}} \sum_{i=1}^{n_{t-k}} m_{\theta}(D_{t-k,i}) &\geq \sum_{k=1}^K\beta_k \frac{1}{n_{t-k}} \sum_{i=1}^{n_{t-k}} m_{\hat{\theta}_t^{\beta}}(D_{t-k,i}) \\ &\geq \sum_{k=1}^K\beta_k \frac{1}{n_{t-k}} \sum_{i=1}^{n_{t-k}} m_{\theta^0}(D_{t-k,i}) - o_P(1),
\end{align*}
where the last inequality comes from the definition of $\hat{\theta}_t^{\beta}$.
Using Theorem~\ref{thm:clt} (without the linear transformation in Lemma~\ref{lemma:dist}) with $\phi(D) = m_{\theta^0}(D)$, we have
\begin{equation*} 
   \sum_{k=1}^K\beta_k \frac{1}{n_{t-k}} \sum_{i=1}^{n_{t-k}} m_{\theta^0}(D_{t-k,i}) - o_P(1) = \mathbb{E}^0[m_{\theta^0}(D)] - o_P(1).
\end{equation*}
Therefore, 
\begin{equation*}
    \{\hat{\theta}_t^{\beta} \in B\} \subset \Bigg\{ \sup_{\theta \in B} \sum_{k=1}^K\beta_k \frac{1}{n_{t-k}} \sum_{i=1}^{n_{t-k}} m_{\theta}(D_{t-k,i}) \geq \mathbb{E}^0[m_{\theta^0}(D)] - o_P(1) \Bigg\}.
\end{equation*}
By the equation \eqref{eq:pf-lemma2}, the probability of the event on the right hand side converges to zero as $n \xrightarrow[]{} \infty$. Similarly, the consistency of $\theta_t - \theta^0 = o_p(1)$ can be proved. 
\end{proof}

\subsection{Proof of Proposition~\ref{prop:consistency}}

\begin{proof}
We first send $m(n), n \xrightarrow[]{} \infty$. Note that 
\begin{equation}\label{eq:opt-quad}
    \beta^*  = \argmin_{\substack{\beta: \beta^{\intercal}1 = 1 \\ \beta\geq 0}} \beta^{\intercal} \left(\Sigma^W + r\cdot I\right) \beta = \argmin_{\substack{\beta: \beta^{\intercal}1 = 1 \\ \beta\geq 0}} \beta^{\intercal} \tilde{\Sigma} \beta,
\end{equation}
where $\Sigma^W$ is defined in \eqref{eq:sigma^w}, $I$ is an identity matrix, and $r = \lim_{n\xrightarrow[]{}\infty} m/n$. From the proof of Theorem~\ref{thm:clt}, we have
\begin{align*}
    \hat{\beta} &= \argmin_{\substack{\beta: \beta^{\intercal}1 = 1 \\ \beta \geq 0}} \frac{1}{T-K}\frac{1}{L} \sum_{t>K}^{T}\sum_{\ell=1}^L \left(\sum_{k=1}^K \beta_k \cdot m\left(\hat{\mathbb{E}}^t[\phi_{\ell}] - \hat{\mathbb{E}}^{t-k}[\phi_{\ell}]\right)\right)^2 \\
    &= \argmin_{\substack{\beta: \beta^{\intercal}1 = 1 \\ \beta \geq 0} } \beta^{\intercal} \left( \frac{1}{T-K}\frac{1}{L} \sum_{t>K}^{T}\sum_{\ell=1}^L  Z_{t}^{\ell}{Z_{t}^{\ell}}^{\intercal} +o_p(1) \right) \beta
\end{align*}
where $Z_{\ell}^t \sim N(0, \Sigma^W + r\cdot I + r\cdot 11^{\intercal})$. Here,  $r \cdot 11^{\intercal}$ comes from having  $\hat{\mathbb{E}}^t$ instead of $\mathbb{E}^t$. Note that for $\ell \neq \ell'$, $Z_{\ell}^t$ and $Z_{\ell'}^t$ are independent which can also be checked in the proof of Theorem~\ref{thm:clt}. 

With fixed $T$ and $L \xrightarrow[]{} \infty$, we have
\begin{equation*}
    \frac{1}{L}\sum_{\ell=1}^L Z_{t}^{\ell}{Z_{t}^{\ell}}^{\intercal} \xrightarrow[]{p} \Sigma^W + r\cdot I + r\cdot 11^{\intercal} + o_p(1).
\end{equation*}
Note that $\beta^{\intercal} (11^{\intercal}) \beta = 1$ with the linearity constraint. Therefore, 
\begin{equation}\label{eq:hat-quad}
    \hat{\beta} = \argmin_{\substack{\beta: \beta^{\intercal}1 = 1 \\ \beta \geq 0}} \beta^{\intercal} \hat{\Tilde{\Sigma}}\beta
\end{equation}
where $\hat{\Tilde{\Sigma}} \xrightarrow[]{p} \Tilde{\Sigma}$. Now let's fixed $L$ and send $T \xrightarrow[]{} \infty$. Note that for each $\ell, i, j$
\begin{align*} 
    & \exists C \text{ such that } \text{Var}(Z_{t,i}^{\ell} Z_{t,j}^{\ell}) \leq C  \text{ for all } t, 
    \\ &\lim_{|t-t'| \xrightarrow[]{}\infty} Cov(Z_{t, i}^{\ell}Z_{t, j}^{\ell}, Z_{t', i}^{\ell} Z_{t', j}^{\ell}) \xrightarrow[]{} 0. 
\end{align*}

Then, by Bernstein's WLLN, we have
\begin{equation*}
    \frac{1}{T-K}\sum_{t>K}^T Z_{t}^{\ell}{Z_{t}^{\ell}}^{\intercal} \xrightarrow[]{p} \Sigma^W + r\cdot I + r\cdot 11^{\intercal} + o_p(1).
\end{equation*}
Then, again we have \eqref{eq:hat-quad} with $\hat{\Tilde{\Sigma}} \xrightarrow[]{p} \Tilde{\Sigma}$.

Let $C: \{\beta: \beta^{\intercal}1 = 1, \beta \geq 0\}$. Define
\begin{equation*}
    \epsilon = \sup_{\beta \in C} |\beta^{\intercal}(\hat{\Tilde{\Sigma}} - \Tilde{\Sigma})\beta|. 
\end{equation*}
Note that
\begin{equation*}
{\beta^*}^{\intercal}\Tilde{\Sigma}\beta^* \leq {\hat{\beta}}^{\intercal} \Tilde{\Sigma} \hat{\beta} \leq \hat{\beta}^{\intercal} \hat{\Tilde{\Sigma}} \hat{\beta} + \epsilon  \leq {\beta^*}^{\intercal} \hat{\Tilde{\Sigma}} \beta^* + \epsilon \leq {\beta^*}^{\intercal}\Tilde{\Sigma} \beta^* + 2\epsilon.
\end{equation*}
With $\hat{\Tilde{\Sigma}} \xrightarrow[]{p} \Sigma$, we have $\epsilon = o_p(1)$ as  $C$ is a compact set. Therefore, we have that 
\begin{equation*}
    \Tilde{\delta}^2(\hat{\beta})  =  {\hat{\beta}}^{\intercal} \Tilde{\Sigma} \hat{\beta} = {\beta^*}^{\intercal}\Tilde{\Sigma}\beta^* + o_p(1) = \Tilde{\delta}^2(\beta^*) + o_p(1).
\end{equation*}
This completes the proof. 
\end{proof}

\end{document}